\documentclass[10pt, conference, letterpaper]{IEEEtran}

\usepackage{amsmath}
\usepackage{amsthm}
\usepackage{amssymb}
\usepackage{graphicx} % support the \includegraphics command and options
\usepackage{subfigure}
\usepackage{url}
\usepackage{booktabs} % for much better looking tables
\usepackage{array} % for better arrays (eg matrices) in maths
\usepackage{verbatim} % adds environment for commenting out blocks of text & for better verbatim
\usepackage{caption}
\usepackage{enumitem}
\usepackage{cite}
\usepackage{balance}

\usepackage{dirtytalk} % thodoris

\newtheorem{theorem}{Theorem}%[section]
\newtheorem{mylemma}{Lemma}
\newtheorem{corollary}{Corollary}

\usepackage{multirow}
\usepackage{color}

\newcommand{\hmmm}[1]{\textcolor[rgb]{0,0,0}{{#1}}}
\newcommand{\newtext}[1]{\textcolor[rgb]{0,0,0}{{#1}}}

\newcommand{\eq}[1]{Eq.~\eqref{#1}}
\usepackage{soul}
\setstcolor{red}
\newcommand{\myitem}[1]{\vspace{0.25\baselineskip}\noindent\textbf{#1}}%\vspace*{0.04in}
\newcommand{\myitemit}[1]{\vspace{0.25\baselineskip}\noindent\textit{#1}}%\vspace*{0.04in}

\definecolor{orange}{rgb}{1,0.5,0}

\usepackage{xspace}

\newcommand{\myquotation}[1]{\begin{center}\textit{``#1''}\end{center}}

\newcommand{\keyfinding}[2]{\myitem{Key finding #1:  \textit{``#2''}}}

\usepackage{algorithm}
\usepackage{algorithmicx}
\usepackage{algpseudocode}

\newcommand\blfootnote[1]{%
  \begingroup
  \renewcommand\thefootnote{}\footnote{#1}%
  \addtocounter{footnote}{-1}%
  \endgroup
}

\title{Fairness in Network-Friendly Recommendations}
\author{Theodoros Giannakas\textsuperscript{$\dagger$}, Pavlos Sermpezis\textsuperscript{$\ddagger$},
Anastasios Giovanidis\textsuperscript{$\star$}, \\
Thrasyvoulos Spyropoulos\textsuperscript{$\dagger$}, George Arvanitakis\textsuperscript{$\ddagger$}
\\~\\
\textsuperscript{$\dagger$}EURECOM,~France; firstname.lastname@eurecom.fr\\
\textsuperscript{$\ddagger$}Aristotle University of Thessaloniki,~Greece; \{sermpezis, garvanitakis\}@csd.auth.gr \\
\textsuperscript{$\star$}Sorbonne University, CNRS, LIP6,~France; anastasios.giovanidis@lip6.fr
}

\begin{document}

% \IEEEoverridecommandlockouts
% \IEEEpubid{\makebox[\columnwidth]{978-1-6654-2263-5/21/\$31.00~\copyright2021 IEEE.\hfill} \hspace{\columnsep}\makebox[\columnwidth]{ }}

\maketitle

% \IEEEpubidadjcol

\pagestyle{plain}

\begin{abstract}
As mobile traffic is dominated by content services (e.g., video), which typically use recommendation systems, the paradigm of network-friendly recommendations (NFR) has been proposed recently to boost the network performance by promoting content that can be efficiently delivered (e.g., cached at the edge). NFR increase the network performance, however, at the cost of being unfair towards certain contents when compared to the standard recommendations. This unfairness is a side effect of NFR that has not been studied in literature. Nevertheless, retaining fairness among contents is a key operational requirement for content providers. This paper is the first to study the fairness in NFR, and design fair-NFR. Specifically, we use a set of metrics that capture different notions of fairness, and study the unfairness created by existing NFR schemes. Our analysis reveals that NFR can be significantly unfair. We identify an inherent trade-off between the network gains achieved by NFR and the resulting unfairness, and derive bounds for this trade-off. We show that existing NFR schemes frequently operate far from the bounds, i.e., there is room for improvement. To this end, we formulate the design of Fair-NFR (i.e., NFR with fairness guarantees compared to the baseline recommendations) as a linear optimization problem. Our results show that the Fair-NFR can achieve high network gains (similar to non-fair-NFR) with little unfairness.
% As mobile traffic is dominated by content services (e.g., video), which typically use recommendation systems, the paradigm of network-friendly recommendations (NFR) has been proposed recently to boost the network performance by promoting content that can be efficiently delivered (e.g., cached at the edge). NFR increase the network performance, however, at the cost of being unfair towards certain contents. This unfairness, as a side effect of NFR, has not been studied in literature. Nevertheless, retaining fairness among contents is a key operational requirement for content providers. This paper is the first to study the fairness in NFR, and design fair-NFR. Specifically, we use a set of metrics that capture different notions of fairness, and study the unfairness created by existing NFR schemes. Our analysis reveals that NFR can be significantly unfair. We identify an inherent trade-off between the network gains achieved by NFR and the resulting unfairness, and derive bounds for this trade-off. We show that existing NFR schemes frequently operate far from the bounds, i.e., there is room for improvement. To this end, we take fairness into account and formulate the design of Fair-NFR as a linear optimization problem. Our results show that the Fair-NFR can achieve high network gains (similar to non-fair-NFR) with little unfairness.
\end{abstract}

\blfootnote{This research is co-financed by Greece and the European Union (European Social Fund-ESF) through the Operational Programme ``Human Resources Development, Education and Lifelong Learning'' in the context of the project ``Reinforcement of Postdoctoral Researchers - 2nd Cycle'' (MIS-5033021) implemented by the State Scholarships Foundation (IKY), and the Operational Program Competitiveness, Entrepreneurship and Innovation under the call RESEARCH-CREATE-INNOVATE (project T2EDK-04937). It is also funded by the ANR (French National Agency of Research) by the ``FairEngine'' project under grant ANR-19-CE25-0011, by the ANR \say{5C-for-5G} project under grant ANR-17-CE25-0001, and the IMT F\&R, \say{Joint Optimization of Mobile Content Caching and Recommendation} project.}

\section{Introduction}\label{sec:introduction}
\myitem{Background.} The paradigm of network-friendly recommendations (NFR) has been very recently proposed as a promising solution for improving the quality and/or the cost of content delivery% (and/or reducing its cost for the network)
%, under peak hours or poor network conditions
~\cite{sch-chants-2016,chatzieleftheriou2017caching,sermpezis2018soft,giannakas-wowmom-2018,kastanakis-cabaret-mecomm-2018,kastanakis2020network,zhu2018coded,chatzieleftheriou2019jointly,costantini2019approximation,garetto2020similarity,qi2018optimizing,giannakas2019order,chatzieleftheriou2019joint,gupta2019effect,lin2018joint,song2018making,lin2019content,giannakas2020soba,sermpezis2019towards,cache-centric-video-recommendation,content-recommendation-swarming}. NFR is based on the fact that content %(in particular, video) 
traffic dominates the mobile traffic today~\cite{cisco2018,ericsson2018} and the majority of content services (online video, radio, social networks, etc.) employ recommendation systems (RS), which heavily affect the user choices and shape the content demand~\cite{RecImpact-IMC10, gomez2016netflix}. The main idea behind NFR is to nudge the recommendations of the RS of the content provider towards content that can be delivered in a ``network-friendly'' way (e.g., cached in the mobile edge~\cite{chatzieleftheriou2017caching,sermpezis2018soft,giannakas-wowmom-2018,kastanakis-cabaret-mecomm-2018,kastanakis2020network,zhu2018coded,chatzieleftheriou2019jointly,costantini2019approximation,garetto2020similarity}, or coded broadcast transmissions~\cite{lin2018joint,song2018making,lin2019content}), thus shaping the user demand in favor of this content.

The NFR paradigm involves three main parties: the network, the users, and the content provider. 
Existing works design NFR schemes that explicitly aim to benefit the \textit{network}. Indeed, the envisioned network gains (lower load, congestion, resources, or costs) have been shown to be very promising~\cite{sch-chants-2016,chatzieleftheriou2017caching,sermpezis2018soft,giannakas-wowmom-2018,kastanakis-cabaret-mecomm-2018,kastanakis2020network,zhu2018coded,lin2018joint,song2018making,qi2018optimizing,chatzieleftheriou2019jointly,giannakas2019order,chatzieleftheriou2019joint,gupta2019effect,lin2019content,costantini2019approximation,garetto2020similarity,giannakas2020soba}. Moreover, the \textit{user} experience can improve as well due to the higher satisfaction from high quality content delivery~\cite{sermpezis2019towards}. Finally, there can be benefits for the \textit{content provider} (e.g., higher user engagement); however, those have only been envisioned as a consequence of the higher user satisfaction, but have not been explicitly studied.

% The NFR paradigm involves three main parties: the network, the users, and the content provider. Previous works have shown that NFR are beneficial for the \textit{network} (lower load, congestion, resources, or costs)~\cite{sch-chants-2016,chatzieleftheriou2017caching,sermpezis2018soft,giannakas-wowmom-2018,kastanakis-cabaret-mecomm-2018,zhu2018coded,lin2018joint,song2018making,qi2018optimizing,chatzieleftheriou2019jointly,giannakas2019order,chatzieleftheriou2019joint,gupta2019effect,lin2019content,costantini2019approximation,garetto2020similarity} and the \textit{user} experience (higher satisfaction from high quality content delivery)~\cite{sermpezis2019towards}. Moreover, the user experience has been explicitly taken into account in NFR schemes, by considering the quality of recommendations (QoR) in the nudged recommendations, e.g., by imposing a minimum threshold in the content similarity~\cite{giannakas-wowmom-2018} or a window of user preferences~\cite{chatzieleftheriou2017caching}. However, the benefits for the \textit{content provider} (e.g., higher user engagement) have been envisioned only as a consequence of the higher user satisfaction, but have not been explicitly studied.

\myitem{The problem: Fairness in NFR.} 
To enable network benefits through NFR, the ``cost'' to be paid by the RS is that NFR (a) nudge the optimal recommendations list provided to users, which may lead to worse user satisfaction, and (b) bias the demand for different contents (by making some contents more and others less popular), which may lead to displeasure from the content owners/producers (e.g., YouTubers). The former (user perspective) has been explicitly taken into account in NFR schemes, by considering the \textit{quality of recommendations (QoR)} in the nudged recommendations, e.g., by imposing a minimum threshold in the content similarity~\cite{giannakas-wowmom-2018} or a window of user preferences~\cite{chatzieleftheriou2017caching}. However, the latter (content provider perspective) has been overlooked in related literature. In fact, the shaping of the content demand relates to the \textit{fairness} of a RS towards the content producers/owners, which is a key requirement for content providers and has attracted a lot of attention recently in the design of RS~\cite{abdollahpouri2019multi,burke2017multisided,burke2018balanced,edizel2020fairecsys,patro2020incremental,pessach2020algorithmic,Sacharidis2019ACA,steck2018calibrated,yang2017measuring,liu2018personalizing}.

% NFR affect the operation of the content provider in two ways: (i) they modify the list of recommendations of the baseline RS (short term effect), and (ii) shape the content demand favoring a set of network-friendly contents (long term effect). While the former relates to the user experience and is accounted in the NFR by the QoR constraints, the latter has been overlooked in related literature. In fact, the shaping of the content demand relates to the \textit{fairness} of a RS towards the content producers/owners, which is a key requirement for content providers and has attracted a lot of attention recently in the design of RS~\cite{abdollahpouri2019multi,burke2017multisided,burke2018balanced,edizel2020fairecsys,patro2020incremental,pessach2020algorithmic,Sacharidis2019ACA,steck2018calibrated,yang2017measuring,liu2018personalizing}.

On one hand, some unfairness due to NFR may be acceptable by the content providers under some conditions (during periods of network congestion, peak hours, etc.), in order to better satisfy the users or increase their engagement by avoiding serving them content in poor quality. However, previous works have not studied \textit{how much unfairness is created by the NFR schemes}, and whether this is acceptable by the content provider. 
%
% \myquotation{How much unfairness is created by the NFR schemes?}
%
On the other hand, a content provider may need to satisfy some explicit fairness requirements for the contents (or, the content producers/owners), e.g., not allow a change in the demand larger than 5\%. Up to now, this is not an option in the existing NFR schemes, since \textit{fairness requirements have not been considered as a design aspect in NFR}. 

% \myquotation{QoR is fair to the user but may not be fair to the content.}

\myitem{Contributions.} Motivated by this gap in literature, this paper is the first to study the aspect of fairness in NFR:
%Motivated by the aforementioned gap in the literature, this paper is the first to study the aspect of fairness in NFR. Specifically, our contributions are:
\begin{itemize}[leftmargin=*,nosep]
\item \textit{\textbf{Fairness characterization.}} %We quantify and study the unfairness in NFR schemes. 
We use metrics that capture different notions of fairness in RS (Section~\ref{sec:preliminaries}), and then quantify the unfairness created in a wide range of representative scenarios and NFR algorithms, and investigate the role of different system parameters (Section~\ref{sec:characterization}).

\item \textit{\textbf{The fairness vs. network gain trade-off.}} We identify an inherent trade-off between the network gains achieved by a NFR scheme and the resulting unfairness. We analytically study this %the networks gains vs. fairness 
trade-off and derive bounds% in closed-form expressions
. We show that existing NFR schemes, frequently operate far from the optimal operating point that is given by the bound (Section~\ref{sec:bounds}).

\item \textit{\textbf{Optimal Fair-NFR.}} We formulate the problem of designing NFR that maximize the network gain, under fairness guarantees compared to the baseline RS%where we model and incorporate the fairness as a constraint
. Through a series of transformations, we show that the problem of optimal fair-NFR can be expressed as a linear program (Section~\ref{sec:design}).

\item \textit{\textbf{The price of fairness.}} Studying the performance of the Fair-NFR scheme shows that %. Our results show that 
\textit{by allowing a little unfairness, high network gains can be achieved}, which is a promising message for the NFR paradigm. A comparison with (non-fair) NFR schemes demonstrates that \textit{the Fair-NFR scheme achieves equal gains with much less unfairness} (Section~\ref{sec:price}).
\end{itemize}

\section{Preliminaries}\label{sec:preliminaries}
\subsection{Network-friendly Recommendations}\label{sec:preliminaries-model}

We consider a content service that has integrated in its (web/mobile) platform a recommendation system (RS). When a user is in the platform and consumes (e.g., watches, listens to, reads, buys) a content, a list of recommendations is presented by the RS suggesting to her to consume another content next. This is a typical scenario for the majority of online video/radio services (e.g., YouTube, Netflix, Spotify), news sites, e-shops and online marketplaces (e.g., Amazon), online social networks (e.g., Facebook, Instagram), etc. In the following, we describe the generic setup considered in NFR; the main notation is summarized in Table~\ref{tab:notation}.

\myitem{Content service.} Assume that the service has a content catalog $\mathcal{K}$ ($|\mathcal{K}|=K$). Users request contents in two ways: (i) directly, e.g., by following an external link or typing the content through a search bar, or (ii) by following one of the recommendations provided by the RS of the service (users typically consume several contents when visiting the service). These are the main types of demand in most content services.

We define the \textit{demand} $p_{i}$ for a content $i$ as the fraction of all requests (i.e., direct and through recommendations) that are for this content; we denote as $\mathbf{p} = [p_{1}, ..., p_{K}]$ the vector with the distribution of total demand for all contents.

\myitem{Network.} We assume that a subset of the content catalog $\mathcal{C}\subset\mathcal{K}$ can be delivered with low cost for the network (and/or in high quality). For instance, in the context of mobile edge caching considered by the majority of related work in NFR~\cite{sch-chants-2016,chatzieleftheriou2017caching,sermpezis2018soft,giannakas-wowmom-2018,zhu2018coded,chatzieleftheriou2019jointly,costantini2019approximation}, the contents in $\mathcal{C}$ are cached in the mobile edge. In this context, and w.l.o.g., we set the cost for delivering contents in $\mathcal{C}$ to zero and the cost for the other contents to $1$. Hence, the cache hit ratio, $CHR = \sum_{i\in\mathcal{C}}p_{i}$, captures the total benefit for the network (i.e., the decrease in the cost by using a cache).

\myitem{Recommendations.} We assume a ``recommendation score'' $u_{ij}$ for every pair of contents $i,j\in\mathcal{K}$, which indicates how good a recommendation for content $j$ after content $i$ is. The score $u_{ij}$ may correspond to the similarity between two contents, or more generally to the relevance of recommending $j$ after $i$ (e.g., capturing from item-item collaborative filtering~\cite{sarwar2001item} to black-box deep learning architectures~\cite{covington2016deep}), and can be the output of any state-of-the-art RS. W.l.o.g, we assume $u_{ij}\in[0,1]$ and higher values denote better recommendations. 

\textit{Baseline RS (BS-RS)} is the standard RS (i.e., non network-friendly) that generates the recommendation scores $u_{ij}$ and is used in production by the content/service provider.
%that is used by the content service/provider
After a user has consumed content $i$, the BS-RS recommends to the user a list $R_{i}^{BS}$ that contains the $N$ contents with the highest recommendation score values $u_{ij}$. %Depending on how the recommendation scores are defined (our model allows any arbitrary definition), the recommendations of the BS-RS can correspond to the top-N contents that are most relevant to the user, or most similar to the currently consumed content, or satisfy the highest diversity in list, etc. 
% More formally the list of recommendations $R_{i}$ of the BS-RS is
% \begin{equation*}
% R_{i}^{BS} = \{S\subset \mathcal{K}: |S|=N \wedge u_{ij} \geq u_{i\ell} ~\forall j\in S,  \ell\in\mathcal{K}\backslash S \}
% \end{equation*}
%\theo{Skeftomai mipws ftiaxname ena mikro sxediagramma. Diladi kati na deixnei oti o CP exei ws input to U kai meta vgazei $R^{BS}$, enw antistoixa na kanoume to idio gia NF-RS, kai meta to network na pairnei ws eisodo to $\mathbf{p^{BS}}$, to U kai to cost vector. btw, gw na to ftiaksw ennow :P, de kserw an thes kai an exoume xwro omws.}

\textit{Network-friendly RS (NF-RS)} is a RS that takes into account the network conditions (e.g., delivery cost~\cite{giannakas-wowmom-2018,sermpezis2019towards}, cached contents~\cite{chatzieleftheriou2019jointly,sermpezis2018soft}, wireless channel~\cite{song2018making,lin2019content}) and provides a list of recommendations $R_{i}^{NF}$ to the user. In general, the lists $R_{i}^{NF}$ can be the same as those of the BS-RS $R_{i}^{BS}$, partially overlap with them, or be totally disjoint sets. Typically, the recommendations of NF-RS tend to (i) include more recommendations to contents that can be delivered in a network-friendly way (e.g., cached contents), while (ii) trying to maintain the quality of recommendations (QoR) by recommending contents with relatively high scores $u_{ij}$. In a simple example, with one user, three contents $a,b,c$ with scores $u_{a}=1$, $u_{b}=0.8$, $u_{c}=0.5$, and a BS-RS recommending only one content $R^{BS}=[a]$. Let only $b,c\in\mathcal{C}$ be cached; then the NF-RS would recommend $R^{NF}=[b]$, because this would bring network gains, and would have QoR$=\frac{u_{b}}{u_{a}}=0.8$, which is higher than if $c$ was recommended instead of $b$ (QoR=0.5).

Finally, the resulting demand $\mathbf{p}$ depends on the underlying RS: a RS that selects more frequently a content $i$ in the recommendation lists, will lead to an increase in the demand $p_{i}$. In the remainder, we denote with a superscript the RS that corresponds to the content demand, e.g.,  $\mathbf{p^{BS}}$ for the BS-RS and $\mathbf{p^{NF}}$ for a NF-RS. The differences between the vectors $\mathbf{p^{BS}}$ and $\mathbf{p^{NF}}$ capture the notion of fairness, which we formally define below.

\begin{table}%[h]
\caption{Important notation.}
\label{tab:notation}
\centering
\begin{tabular}{|ll|}
\hline
$\mathcal{K}$   & content catalog  ($|\mathcal{K}|=K$) \\
\hline
$p_{i}$  & demand for content $i$;  $\mathbf{p} = \left[p_{1}, ..., p_{K}\right]$ and $\sum_{i\in\mathcal{K}}p_{i}$=1\\
\hline
$\mathcal{C}$   & set of cached contents ($|\mathcal{C}|=C$) \\
\hline
$CHR$ & cache hit rate, $CHR = \textstyle \sum_{i\in\mathcal{C}} p_{i}$\\
\hline
$u_{ij}$ & recommendation score, $u_{ij}\in[0,1]$\\
\hline
$R_{i}$ & list of recommendations after content $i$\\
\hline
$N$ & number of recommendations.\\
\hline
$\alpha$ & probability a user to follow a recommendation\\
\hline
$p^{(d)}_{i}$  & probability that a ``direct request'' is for content $i\in\mathcal{K}$;\\&$\mathbf{p^{(d)}} = \left[p^{(d)}_{1}, ..., p^{(d)}_{K}\right]$ and  $\sum_{i\in\mathcal{K}}p^{(d)}_{i}=1$\\
\hline
$G$ & network gain, $G = CHR^{NF} - CHR^{BS}$\\
% \hline
% {XXX}&{XXXXXXXXXXXXXXXXXX}\\
% \hline
% {$p_{(r)ij}$} & probability that a ``recommended request'' after content $i$ is for\\
% {} & content $j$% $i,j\in\mathcal{K}$
% ; with $p_{(r)ij}> 0$ iff $j\in R_{i}$, and $\sum_{j\in\mathcal{K}}p_{(r)ij}=1$, \\  
% {} & $\forall i\in\mathcal{K}$, and $\mathbf{P_{(r)}}=\{p_{(r)ij}\}$ the $K\times K$ matrix\\
% \hline
% \\
\hline
\end{tabular}
\end{table}

\subsection{Fairness Definition}\label{sec:fairness-definition}
Content providers aim to satisfy two parties, their users (consumers) and the content owners (producers)~\cite{abdollahpouri2019multi,mehrotra2018towards,burke2017multisided,burke2018balanced,Sacharidis2019ACA}, while at the same time maximizing their own utility (e.g., revenue)~\cite{abdollahpouri2019multi}. In general, the goal of a \textit{fair RS} is to strike a balance between \textit{utility} and \textit{satisfaction} of the involved parties~\cite{Sacharidis2019ACA,mehrotra2018towards}.

In the context of NFR, %which aim to maximize the network gains, 
user satisfaction is taken into account with the concept of quality of recommendations (QoR). However, the content owner/producer satisfaction, which is identified as a key component in the design of fair RS, especially in multistakeholder settings~\cite{abdollahpouri2019multi,burke2017multisided}, has been neglected in previous works in NFR. Hence, we focus on the need of the content provider to satisfy content owners/producers, by providing recommendations that are fair with respect to them %have fairness in its recommendations for the content owners/producers 
(which in literature is referred to also as \textit{p-fairness}~\cite{abdollahpouri2019multi,burke2017multisided}).

Fairness in RS can be defined in several ways~\cite{abdollahpouri2019multi,burke2017multisided,burke2018balanced,edizel2020fairecsys,patro2020incremental,pessach2020algorithmic,Sacharidis2019ACA,steck2018calibrated,yang2017measuring}, depending on the system, the involved parties, the needs of the content provider, etc. 
The fairness of a RS can be measured with respect to the recommendations of \textit{a fair RS}. In our setting, where the goal is to quantify the (un)fairness of NFR, this fair RS is by convention the BS-RS (i.e., any standard RS) and the fairness captures the deviation in the total demand $\mathbf{p}$ created by the NF-RS. Thus, a generic measure $F$ can be used:
%In our setting, the fairness of a NF-RS (compared to a BS-RS) can be captured by a generic measure $F$ of the difference between the total demand under the BS-RS and the NF-RS:
\begin{equation}
F = f(\mathbf{p^{BS}}, \mathbf{p^{NF}})   
\end{equation}
In general, the function $f$ can be defined at will according to the use case or requirements of the content provider. For example,~\cite{Sacharidis2019ACA} suggests that $f$ can be any probability divergence measure. Different measures $f$ can capture different notions of fairness. In this paper, we consider the three fairness measures that are most commonly used in literature and practice\footnote{Note that, since we aim to capture the fairness in recommendations, we use metrics from the RS field. Other fairness measures from other fields, e.g., resource allocation~\cite{lan2010axiomatic}, would be less relevant.}:
% The function $f$, as a measure of (un)fairness, can be defined at will, e.g., according to the use case or requirements of the content provider. For example,~\cite{Sacharidis2019ACA} suggests that $f$ can be any probability divergence measure% (e.g., Kullback–Leibler, total variation, Bregman)
% . In this paper, we consider the following three fairness measures that are the most common in related literature and capture different notions of fairness:

\begin{description}
\item[F-max] $F_{max} = \max_{i\in\mathcal{K}}|p_{i}^{NF}-p_{i}^{BS}|$ relates to the \textit{individual fairness}~\cite{pessach2020algorithmic} and accounts for the ``worst case'', i.e., no content will have a demand difference larger than $F_{max}$.

\item[F-tv] $F_{tv} = \frac{1}{2}\cdot \sum_{i\in\mathcal{K}}|p_{i}^{NF}-p_{i}^{BS}|$ is the \textit{total variation distance} between the two distributions, i.e., the average (absolute) change in the content demand~\cite{patro2020incremental}. It allows more flexibility than $F_{max}$ in shaping the demand, since it does not impose a constraint for every single content; e.g., a large difference in a content demand can be compensated by small demand differences in other contents.

\item[F-kl] $F_{kl} =  \sum_{i\in\mathcal{K}}p_{i}^{BS}\cdot \log \left(\frac{p_{i}^{BS}}{p_{i}^{NF}}\right)$ is the Kullback–Leibler (KL) divergence, a widely used measure for the difference between distributions, and commonly used to quantify fairness in RSs~\cite{Sacharidis2019ACA, steck2018calibrated, yang2017measuring}. $F_{kl}$ is more sensitive to changes in contents with lower demand, e.g., an increase $\Delta p$ in the demand $p_{i}^{BS}$ leads to a higher increase in $F_{kl}$ when $p_{i}^{BS}$ is small~\cite{steck2018calibrated}.
\end{description}
\textit{Remark:} Note that $F_{max}, F_{tv}\in[0,1]$, whereas $F_{kl}\in[0,\infty]$ ($F_{kl}\rightarrow\infty$ when $p_{i}^{NF}=0$ and $p_{i}^{BS}\neq0$). For the sake of presentation, in the results we normalize the values of $F_{kl}$ so that it takes values in $[0,1]$ and is comparable with the other metrics. In particular, we use the smoothed version of~\cite{Sacharidis2019ACA, steck2018calibrated}, where we substitute $p_{i}^{NF} \rightarrow (1-w)\cdot p_{i}^{NF} + w\cdot p_{i}^{BS}$, with $w=0.01$ and normalize with its upper bound $\log\frac{1}{w}$; i.e.,
\begin{equation*}
\textstyle
F_{kl} =  \frac{1}{\log\frac{1}{w}}\cdot \sum_{i\in\mathcal{K}}p_{i}^{BS}\cdot \log \left(\frac{p_{i}^{BS}}{(1-w)\cdot p_{i}^{NF} + w\cdot p_{i}^{BS}}\right)    
\end{equation*}

\hmmm{The above metrics reflect different notions of fairness and requirements of the content provider. In general, it is not possible to satisfy all notions of fairness at the same time~\cite{pessach2020algorithmic}. In this paper, we consider all these metrics, and study their characteristics in relation to NFR (Sections~\ref{sec:characterization} and~\ref{sec:bounds}) and take them into account in the design of fair NF-RS (Section~\ref{sec:design}).}

\myitem{QoR vs. fairness.} \newtext{As a remark, we stress that the notions of QoR (considered in previous works) and fairness (not considered before) describe \emph{different} quantities in NFR; the former relates to the satisfaction of the users/consumers, and the latter to the satisfaction of the content owners/producers. The following example demonstrates this distinction: Let two users and three contents $a,b,c$ with scores $u_{a}=1$, $u_{b}=0.8$, $u_{c}=0.8$ (same for both users), and ${b,c}\in\mathcal{C}$, i.e., are cached. The BS-RS recommends content $a$, with the highest score $u$, to both users. Let's assume two NF-RS that nudge the BS-RS recommendations towards cached contents: A NF-RS recommends $b$ to both users, and another NF-RS recommends $b$ to the first user and $c$ to the second user. Since, $u_{b}=u_{c}$ the QoR in both NF-RS is the same. However, the former NF-RS is less fair, e.g., in terms of \textit{F-max}, since it increases twice the demand for content $b$ compared to the latter NF-RS.}

\section{Characterization of Unfairness in NFR}\label{sec:characterization}
In this section, we aim to understand the (un)fairness $F$ in NFR. To this end, we employ an empirical approach where we (i) consider a wide range of scenarios, (ii) apply the BS-RS and different NF-RS algorithms that have been proposed in previous works, and (iii) calculate the resulting content demand and its unfairness (Section~\ref{sec:sim-setup}). We analyze the results to investigate whether existing NFR schemes create unfairness, and what are the key factors that cause it (Section~\ref{sec:characterization-analysis}).

% In this section, we aim to understand the (un)fairness in NFR. First, we investigate whether existing NFR schemes create unfairness, and what are the key factors that cause it (Section~\ref{sec:characterization-analysis}). We find that there exists an inherent trade-off between the network gains that can be achieved by a NF-RS algorithm and the unfairness it creates. We derive analytic bounds (closed form expressions) for the minimum possible unfairness as a function of the network gains under any NFR scheme, and show that frequently the existing NF-RS algorithms operate far from the bounds (Section~\ref{sec:bounds}). 

%\hmmm{We show that frequently the existing NF-RS algorithms operate far from this bound, and shed light on why this happens and how it can be fixed (Section~\ref{sec:similar-different-performance})}. 

\subsection{Simulation Setup}\label{sec:sim-setup}
% In the characterization of the fairness in NFR (Section~\ref{sec:characterization}) and the evaluation of the fair NF-RS  (Section~\ref{sec:price}), we use a simulation setup where we (i) consider different content service and network scenarios, (ii) apply the BS-RS and different NF-RS algorithms that have been proposed in previous works, (iii) simulate the content demand, and (iv) calculate the resulting unfairness and network gains. 

\myitem{Content catalogs.} We consider content catalogs and matrices $\mathbf{U}=\{u_{ij}\}$ extracted from two datasets of real services:

\myitemit{Last.fm.} We use a dataset from the Last.fm database~\cite{lastfm-related-content-dataset}, where we applied the ``getSimilar'' method to the content IDs' and populate the matrix $\mathbf{U}$. As the resulting $\mathbf{U}$ matrix is quite sparse, for the purpose of demonstration, we keep the largest component of the underlying graph, and round to $u_{ij} = 1$ the values above a threshold of $0.1$.

\myitemit{MovieLens.} We use the Movielens movies-rating dataset~\cite{movielens-related-dataset}, containing $69162$ ratings (0 to 5 stars) of $671$ users for $9066$ movies. We apply an item-to-item collaborative filtering (using 10 most similar items) to extract the missing user ratings, and then use the cosine distance to calculate the similarity for each pair of contents. We set $u_{ij}=1$ for contents with cosine distance larger than $0.6$, and $0$ otherwise.

% \myitemit{YouTube FR.} We used the crawler of~\cite{kastanakis-cabaret-mecomm-2018} and collected a dataset from YouTube in our country. We considered 11 of the most popular videos on a given day, and did a breadth-first-search (up to depth 2) on the lists of related videos (max 50 per video) offered by the YouTube API~\cite{youtube-api}. We built the matrix $\mathbf{U}\in\{0,1\}$ from the collected video relations by setting $u_{ij} = 1$ if the content $j$ is one or two hops away from $i$ through the related list of $i$.

% \myitemit{Synthetic.} We genenerate an \emph{Poisson random} graph of content relations $K = 1000$ nodes, where each content/node has on average 8 neighbors.

\myitem{Caching.} We consider cache sizes $C\in \{5,10,20\}$, with a popularity-based caching policy, i.e., the cache contains the $C$ contents with the highest demand under the BS-RS ($\mathbf{p^{BS}}$).
% \begin{equation*}
% \mathcal{C} = \{S\subset \mathcal{K}: |S|=C \wedge p^{BS}_{i} > p^{BS}_{j} ~\forall i\in S,  j\in\mathcal{K}\backslash S \}
% \end{equation*}

\myitem{Content demand.} Similarly to previous works~\cite{chatzieleftheriou2019jointly,giannakas-wowmom-2018,giannakas2019order,zhu2018coded,lin2018joint,lin2019content}, we assume that a user follows a recommendation with probability $\alpha$, or directly requests a content with probability $1-\alpha$. We set $\alpha \in \{0.5, 0.8, 0.99\}$, to capture the behavior reported for YouTube ($\alpha$=0.5)~\cite{RecImpact-IMC10} and Netflix~($\alpha$=0.8)~\cite{gomez2016netflix} and an extreme value where users follow almost always recommendations ($\alpha$=0.99), e.g. as in YouTube autoplay or online radio services like Last.fm, Jango, etc..%~\theo{1 is maybe a bit dangerous, due to the reasons I mention in the opt section, but here since we simply evaluate other works, i think it is fine}\pavlos{Yes, I know that a=1 is problematic in the theoretical framework. In fact, in sims I used 0.99. However, I think here since we refer only to sims, a=1 is ok}

We assume that direct requests for different contents follow a Zipf distribution with exponent $s$, where we used a typical scenario with $s=1$~\cite{RecImpact-IMC10} and an extreme scenario with $s=0$ (i.e., uniform distribution). We denote the distribution of direct requests as $\mathbf{p^{(d)}}$.

\begin{table}
\centering
% \caption{Parameters for the 1296 simulation scenarios}\theo{proteinw na fainetai to 1296}
\caption{Parameters of the simulation scenarios (in total, all their combinations give 1296 scenarios)}
\label{tab:sim-paramaters}
\begin{tabular}{|l|l|}
% \hline
\toprule
$\mathbf{U}$: last.fm ($K=757$), Movielens ($K=1060$)  &
$N\in \{2,5,10\}$ \\
$\mathbf{p^{(d)}}\sim \{Zipf(s=1), uniform\}$ &
$C\in\{5,10, 20\}$ \\
$\alpha\in\{0.5, 0.8, 0.99\}$ &
$q\in \{0.5, 0.8, 0.9\}$\\
$W_{BFS}\in \{N, 2N\}$ & 
$D_{BFS}\in\{1,2\}$\\
% \hline
\bottomrule
\end{tabular}
\end{table}

\myitem{NF-RS algorithms.} Several NFR variants have been proposed% in previous works
. To avoid restricting our study to a single algorithm or setup, we consider three representative NF-RS algorithms. % that span a large range of the considered system setups, user models, and optimization approaches.

\myitemit{Greedy NF-RS} includes in each recommendation list $R_{i}$ as many cached contents %$j\in\mathcal{C}$ 
as possible, without violating a minimum QoR threshold $q$. %If the QoR cannot be satisfied with cached contents only, then the remaining recommendations are for the (non-cached) contents $j$ with the highest recommendation scores $u_{ij}$. 
It aims to maximize the CHR by considering every request independently (without taking into account the long term performance). It can be seen as a simplified version of only the recommendation part of the CawR algorithm~\cite{chatzieleftheriou2019jointly} (with the cache assumed already filled)\footnote{We note that CawR~\cite{chatzieleftheriou2019jointly} optimizes at the same time the caching and recommendation policies. Since the scope of this paper is on the fairness of the recommendations in NF-RS, we focus on the resulting recommendations of NF-RS algorithms, given a pre-filled cache\hmmm{; we discuss implications of joint NF-RS and caching policy optimization algorithms in Section~\ref{sec:conclusion}}.}, or the ``Myopic'' version of CARS~\cite{giannakas-wowmom-2018}. 

\myitemit{Multi-step NF-RS}~\cite{giannakas2019order} is an algorithm that includes in each recommendation list $R_{i}$ a set of contents that satisfy a QoR constraint (similarly to the Greedy NF-RS) and maximizes the network gains in the long term, i.e., by taking into account requests made directly and through recommendations, and the probability $\alpha$. It returns the optimal solution in our model setup under \emph{no fairness} requirements.

\myitemit{CABaRet}~\cite{kastanakis-cabaret-mecomm-2018} follows a different approach, by leveraging the BS-RS and assuming no explicit knowledge on the scores $u_{ij}$. For each content $i$, it does a breadth-first search (BFS) starting from the list $R_{i}^{BS}$ (depth 1), and then to the lists $R_{j}^{BS}$, $\forall j \in R_{i}^{BS}$, (depth 2), and so on. It returns a recommendation list that contains the cached contents found in the BFS and, if needed, fills the list with the initial recommendations $R_{i}^{BS}$. %It selects the width $W_{BFS}$ and depth $D_{BFS}$ of the BFS to achieve the desired QoR. 

In all cases recommendation lists are of size $N\in\{2,5,10\}$.

\myitem{Quality of Recommendations (QoR).} In the {Greedy} and {Multi-step NF-RS}, the QoR constraint is explicitly defined as a fraction of the recommendation quality of the BS-RS by a parameter $q\in[0,1]$, i.e., $\sum_{j\in R_{i}^{NF}} u_{ij} \geq q\cdot \sum_{j\in R_{i}^{BS}} u_{ij}$~\cite{giannakas-wowmom-2018,giannakas2019order}. In {CABaRet}, the QoR is implicitly determined by the width $W_{BFS}$ and depth $D_{BFS}$ parameters of the BFS~\cite{kastanakis-cabaret-mecomm-2018}. In our simulations, we consider values $q\in\{0.5, 0.8, 0.9\}$, and $W_{BFS}\in \{N, 2 N\}$ and $D_{BFS} = \{1,2\}$.

Table~\ref{tab:sim-paramaters} summarizes the parameters of the considered scenarios. In total, we simulated 1296 scenarios, accounting \emph{for all the combinations} of the parameters.

\subsection{Unfairness in NFR}\label{sec:characterization-analysis}

In the scenarios we simulate, we calculate the content demand under the BS-RS ($\mathbf{p^{BS}}$) and the different NF-RS algorithms ($\mathbf{p^{NF}}$), and then the resulting unfairness captured with the metrics $F(\mathbf{p^{NF}},\mathbf{p^{BS}})$ defined in Section~\ref{sec:fairness-definition}. 

\myitem{Unfairness in existing NF-RS algorithms.} We first quantify the unfairness created by existing NF-RS algorithms. Figure~\ref{fig:cdf-fairness} presents the CDF of the values of the fairness metrics $F$ among all the scenarios we tested\hmmm{; large values of the fairness metric $F$ denote more unfair systems (e.g., the system is fair for $F$=0 and very unfair for $F$=1)}. We see that \textit{the NF-RS algorithms create unfairness, which is very high in several cases} (we remind that the presented $F$ metrics take values in $[0,1]$). Moreover, comparing the curves of the different metrics (or, notions) of fairness, we can see that  $F_{max}$ that captures the individual fairness takes lower values, whereas $F_{tv}$ that is averaged over all contents takes the highest values (even up to $1$). The CDF of $F_{kl}$, which considers all contents while also giving emphasis on individual contents whose demand deviates a lot from $\mathbf{p^{BS}}$, lies between the other two metrics.

\begin{figure}%
\centering
\subfigure[b][CDF of unfairness]{
\begin{minipage}[b]{0.45\columnwidth}
\includegraphics[width=1\columnwidth]{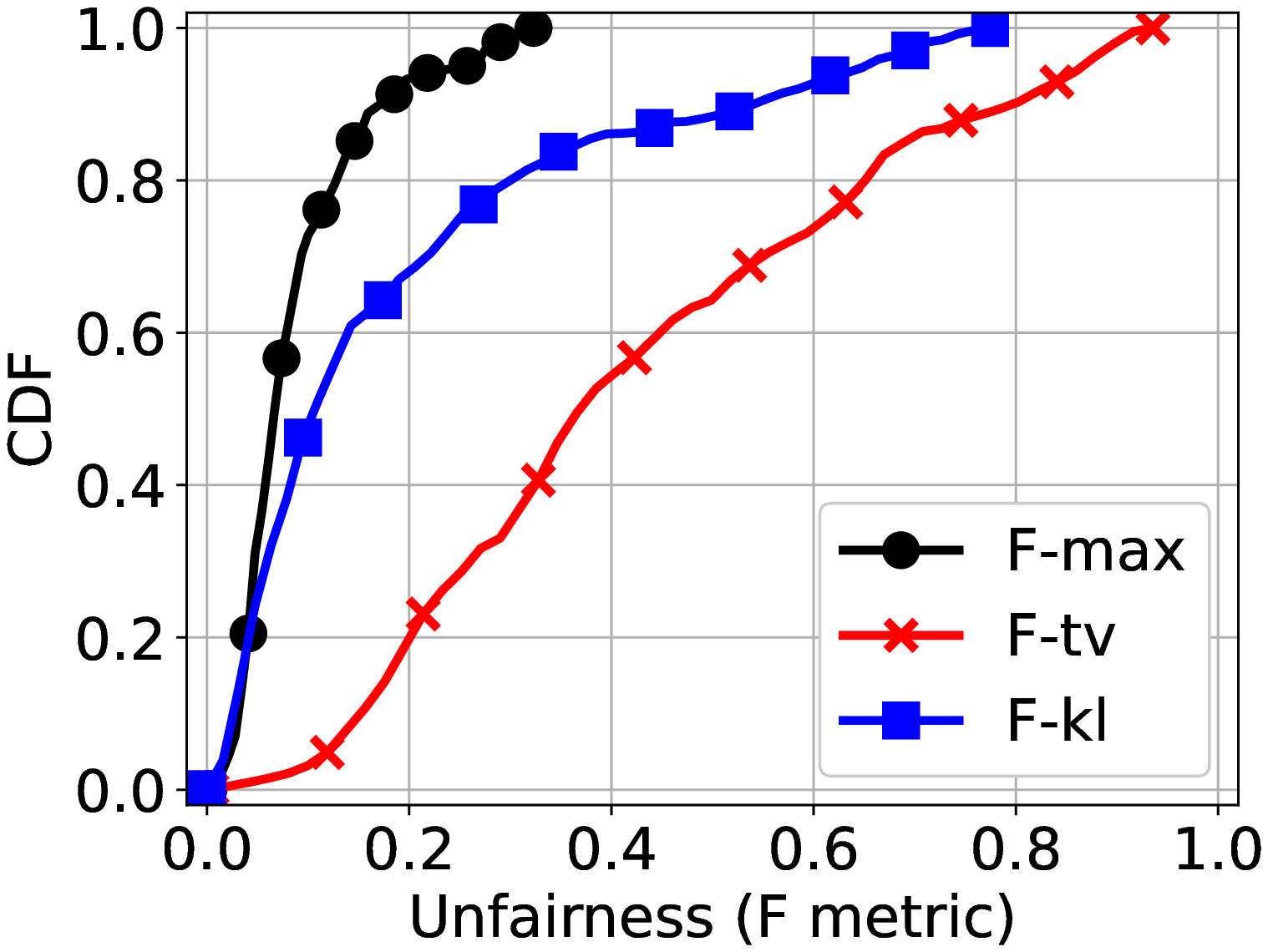}
\end{minipage}
\label{fig:cdf-fairness}
}
\centering
\subfigure[b][Unfairness vs. QoR]{
\begin{minipage}[b]{0.45\columnwidth}
\includegraphics[width=1\columnwidth]{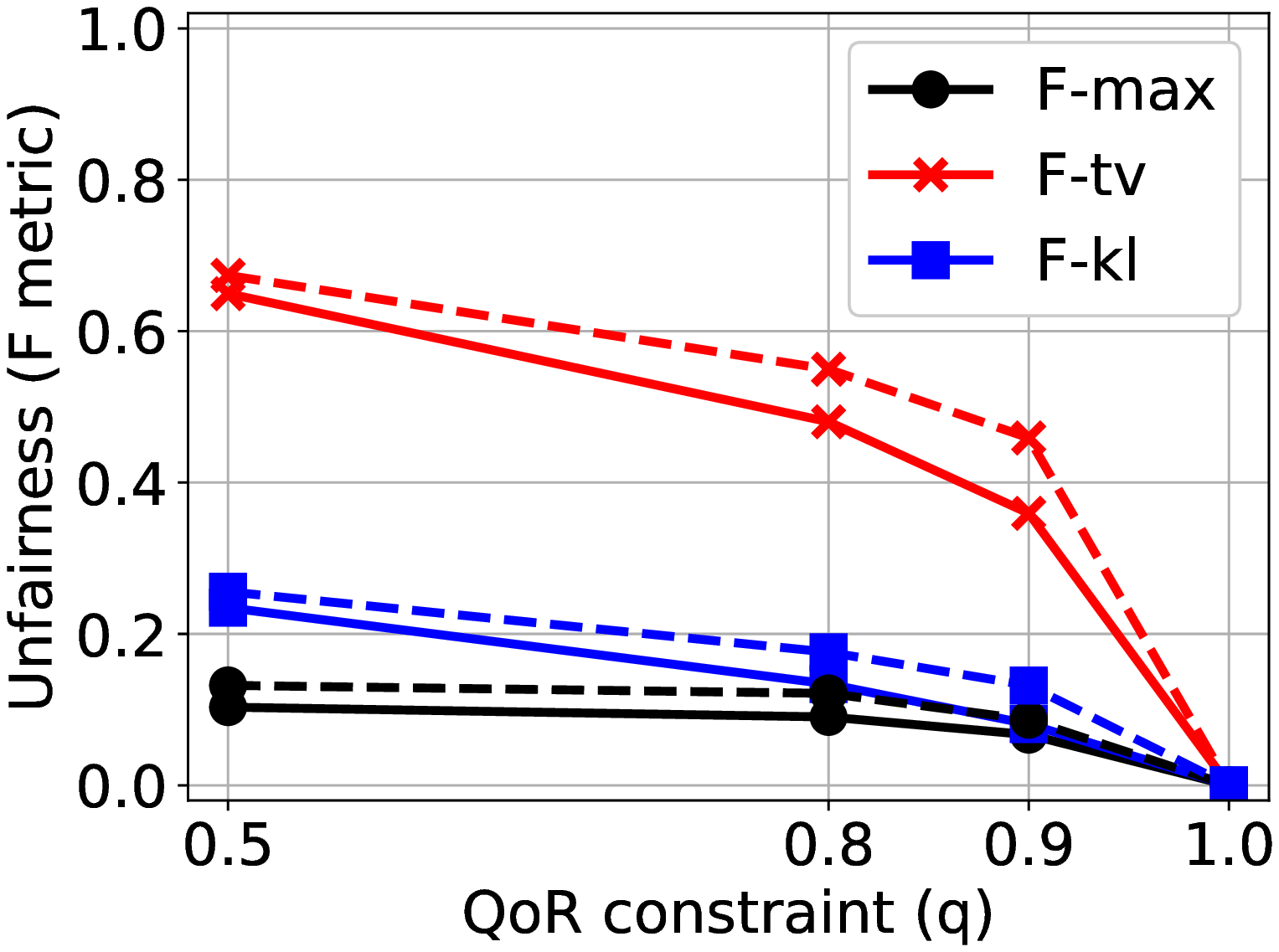}
\end{minipage}
\label{fig:fairness-vs-qor}
}
% \subfigure[b][Unfairness vs. system parameters]{
% \begin{minipage}[b]{0.45\columnwidth}
% ~\vspace{0.2cm}
% {\footnotesize
% \begin{tabular}{c|ccc}
% {} & $F_{max}$ & $F_{tv}$ & $F_{kl}$\\
% \hline
% $q \nearrow $ & $\searrow$ & $\searrow$ & $\searrow$ \\
% $\alpha \nearrow $ & $\nearrow (\rho=0.46)$ & $\nearrow (\rho=0.81)$ & $\nearrow (\rho=0.73)$ \\
% $N \nearrow $ & $\searrow$ & $\searrow$ & $\searrow$ \\
% $C \nearrow $ & $-$ & $-$ &$-$ \\
% \multicolumn{4}{c}{}
% \end{tabular}%
% }%
% \end{minipage}
% \label{fig:fairness-vs-parameters-table}
% }
\caption{(a) CDF of unfairness created by the NF-RS algorithms in all scenarios (Table~\ref{tab:sim-paramaters}). (b) Unfairness vs. QoR, in Movielens scenarios with $\alpha$=0.8, $N$=5, $C$=10, uniform $\mathbf{p^{(d)}}$, Greedy (continuous lines) and Multi-step (dashed lines) NF-RS.}
\end{figure}

\begin{figure*}
\centering
\subfigure[$F_{max}$]{\includegraphics[width=0.49\columnwidth]{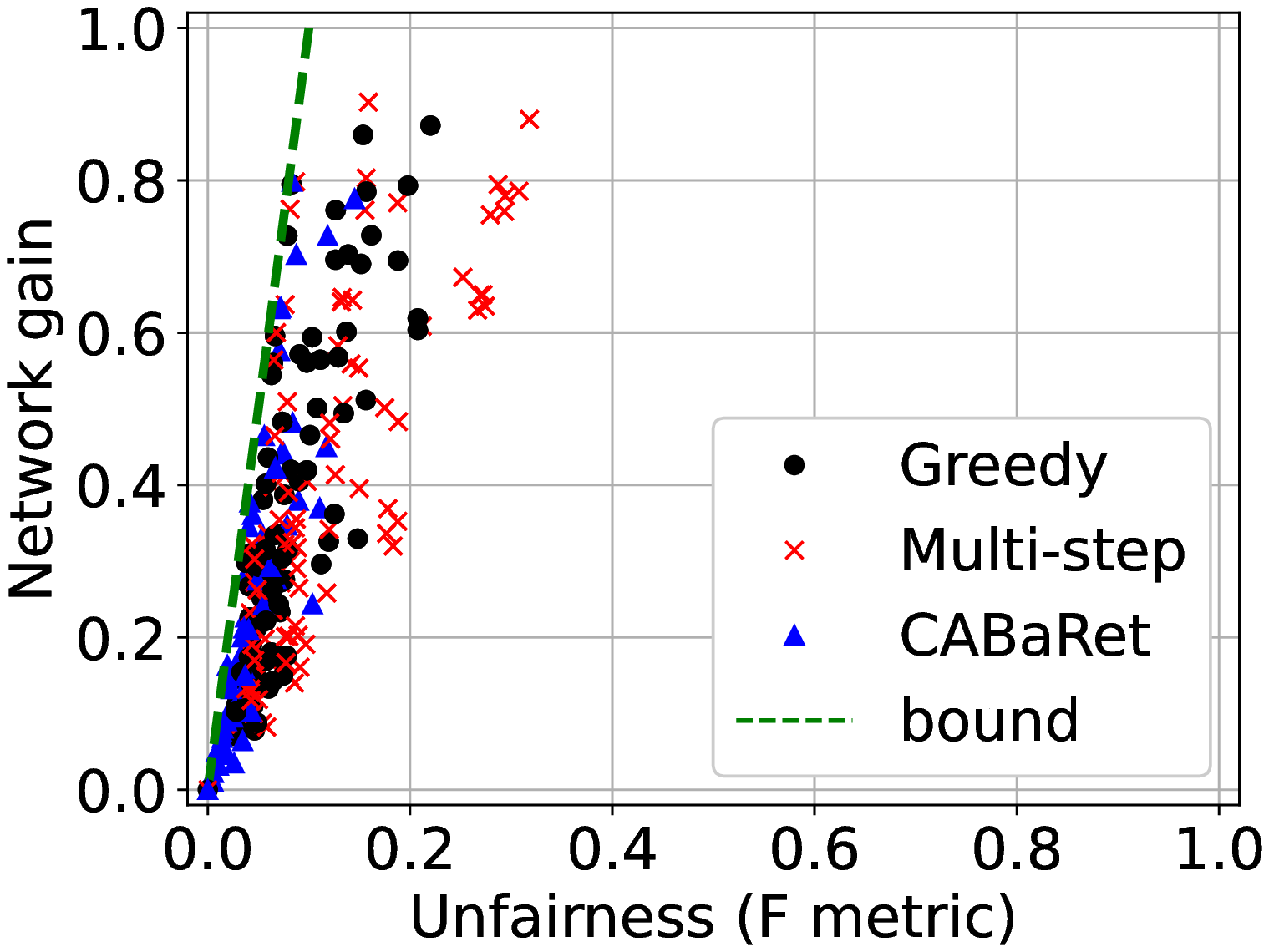}\label{fig:fairness-vs-gain-max}}
\subfigure[$F_{tv}$]{\includegraphics[width=0.49\columnwidth]{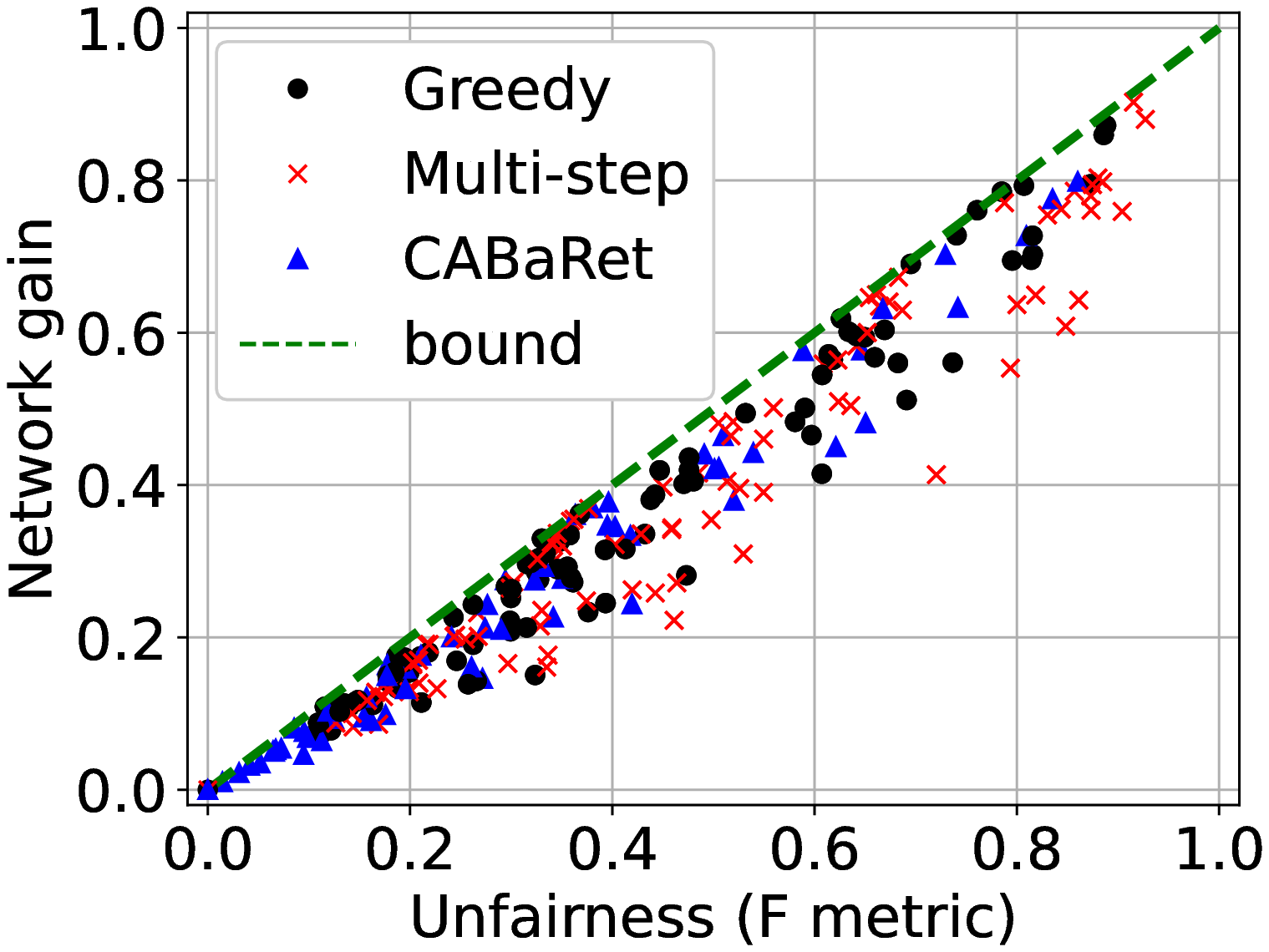}\label{fig:fairness-vs-gain-avg}}
\subfigure[$F_{kl}$]{\includegraphics[width=0.49\columnwidth]{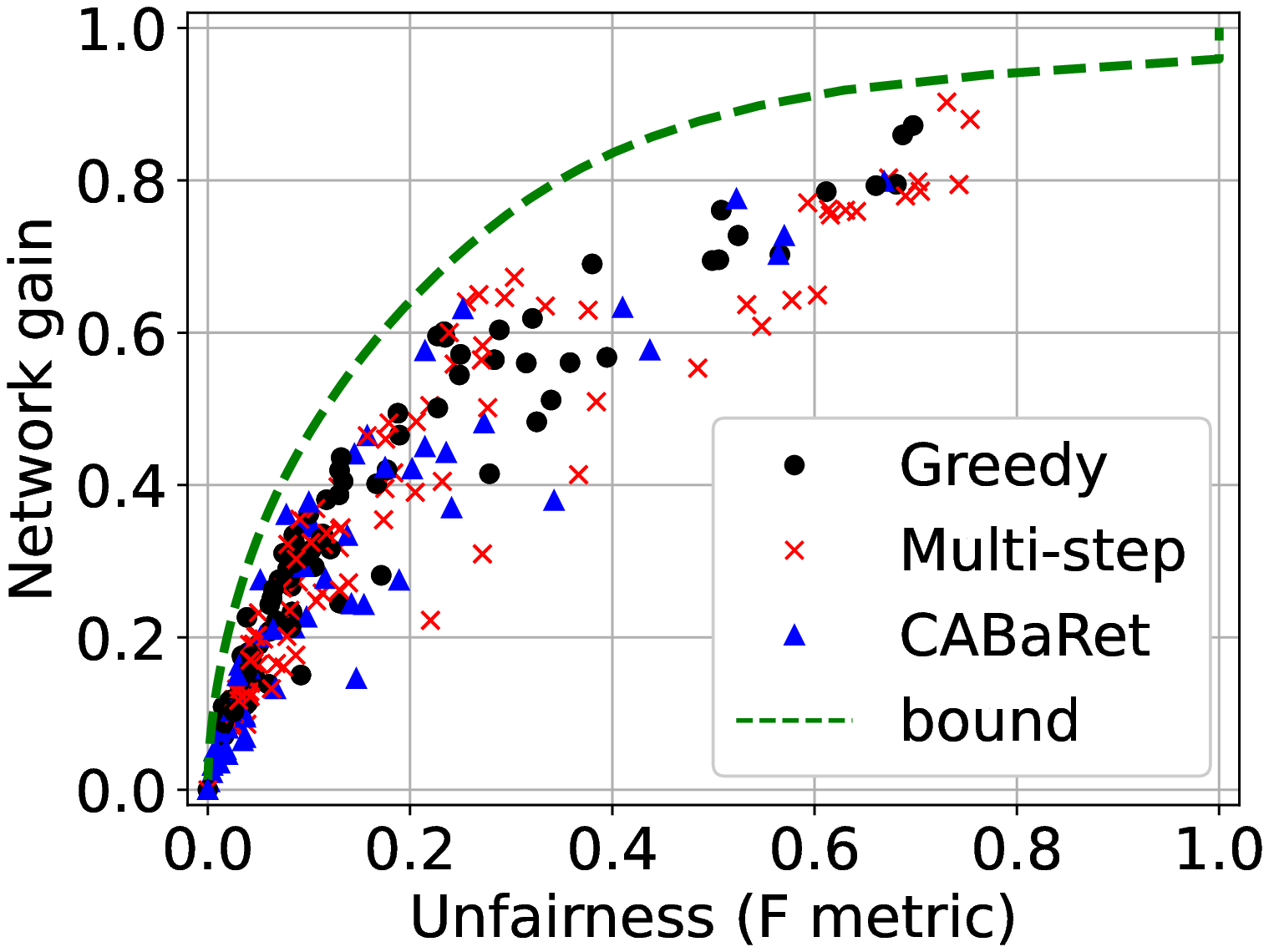}\label{fig:fairness-vs-gain-kl}}
\subfigure[CDF of distance from bounds]{\includegraphics[width=0.49\columnwidth]{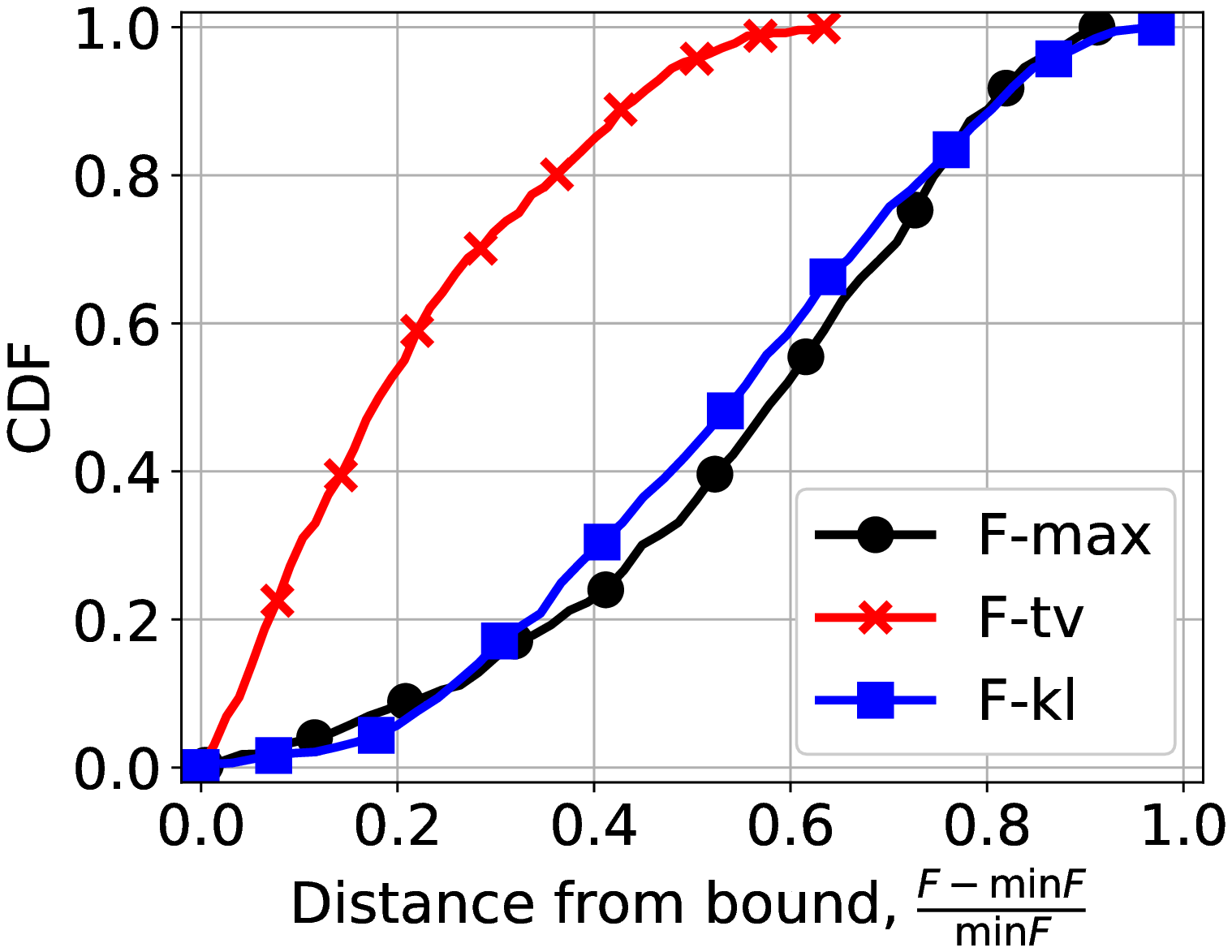}\label{fig:cdf-distance-from-bound}}
\caption{\hmmm{(a),(b),(c):} Network gain $G$ ($y$-axis) vs. fairness metric $F$ ($x$-axis) in all scenarios of Table~\ref{tab:sim-paramaters} with $C=10$, and under the \textit{Greedy} (circles), \textit{Multi-step} (crosses), \textit{CABaRet} (triangles) NF-RS. The bounds are denoted with dashed lines. \hmmm{(d): CDF of the relative distance $\frac{F-\min F}{\min F}$ of the operating points $(F,G)$ of the NF-RS algorithms from the corresponding bounds $(\min F, G)$.}}
\label{fig:fairness-vs-gains}
\end{figure*}

\myitem{The role of the QoR.} Figure~\ref{fig:fairness-vs-qor} presents the resulting unfairness ($y$-axis) by applying the Greedy NF-RS (continuous lines) and the Multi-step NF-RS (dashed lines) in scenarios with different QoR constraints $q$ ($x$-axis). It is clearly seen that as the QoR constraint becomes looser (lower values in $x$-axis) the unfairness of the system increases (higher values in $y$-axis). This is due to the fact that by relaxing the required QoR, the NF-RS has more flexibility in changing the recommendation lists, and consequently this nudges the content demand $\mathbf{p^{NF}}$ farther from $\mathbf{p^{BS}}$. However, it is interesting to note the change in unfairness is not linear to the QoR, but is rather a concave function that decreases more steeply for higher values of QoR. In fact, a key observation in Fig.~\ref{fig:fairness-vs-qor} is that even a small decrease in QoR from the maximum value that corresponds to the BS-RS ($q=1$), can already lead to significant unfairness (e.g., see the $F$ values for $q=0.9$). This finding (similar behavior holds in all the scenarios we tested) highlights the following insight, which is a main motivation for this paper:
\myquotation{The QoR constraint commonly used in NF-RS to satisfy the users, may not suffice to (implicitly) impose fairness for the content provider as well. To account for fairness, one needs to explicitly take it into account when designing a NF-RS.}
%
%Moreover, we can see that for lower values of QoR the fairness increases more slowly; for instance, in Fig.~\ref{fig:fairness-vs-qor} the $F_{max}$ remains almost the same for $q<0.8$. Similar findings on the fairness vs. QoR behavior hold in all the scenarios we tested.

\begin{table}[b]
\centering
\caption{Relation between system parameters, fairness $F$ and network gain $G$: monotonicity and correlation ($\rho$).}
\label{tab:fairness-vs-param}
\begin{tabular}{c|ccc|c}
{} & $F_{max}$ & $F_{tv}$ & $F_{kl}$ & $G$\\
\hline
$q \nearrow $ & $\searrow$ ($\rho$=-0.40) & $\searrow$ ($\rho$=-0.38) & $\searrow$ ($\rho$=-0.30)  & $\searrow$ ($\rho$=-0.42) \\
% $q \nearrow $ & $\searrow$ & $\searrow$ & $\searrow$ \\
% {} & $\rho$=-0.40 & $\rho$=-0.38 & $\rho$=-0.38 \\
\hline
$\alpha \nearrow $ & $\nearrow$ ($\rho$=0.46) & $\nearrow$ ($\rho$=0.81) & $\nearrow$ ($\rho$=0.75) & $\nearrow$ ($\rho$=0.69)  \\
% $\alpha \nearrow $ & $\nearrow$ & $\nearrow$ & $\nearrow$ \\
% {} & $\rho$=0.46 & $\rho$=0.81 & $\rho$=0.73 \\
\hline
$N \nearrow $ & $\searrow$ ($\rho$=-0.47) & $\searrow$ ($\rho$=-0.13) & $\searrow$ ($\rho$=-0.16) & $\searrow$ ($\rho$=-0.20)\\
% $N \nearrow $ & $\searrow$ & $\searrow$ & $\searrow$ \\
% {} & $\rho$=-0.47 & $\rho$=-0.13 & $\rho$=-0.21 \\
\hline
$C \nearrow $ & $\searrow$ ($\rho$=-0.13) & $-$ ($\rho$=0.06) &$-$ ($\rho$=0.03) & $\nearrow$ ($\rho$=0.18) 
% $C \nearrow $ & $\searrow$ & $-$ &$-$ \\
% {} & $\rho$=-0.13 & $\rho$=0.06 &$\rho$=0.02
\end{tabular}
\end{table}

\myitem{The role of the NF-RS algorithm and the system parameters.} Comparing the curves of the two NF-RS algorithms in Fig.~\ref{fig:fairness-vs-qor}, we see that the unfairness introduced by the Multi-step NF-RS is higher than the Greedy NF-RS, under any fairness metric $F$. This finding holds in all (for $F_{tv}$, $F_{kl}$) and in 95\% (for $F_{max}$) of the scenarios we tested, and is due to the fact that the Multi-step NF-RS, by accounting the long term behavior, can shape in a larger degree than the Greedy NF-RS (or other heuristics) the content demand under the same QoR constraint. Hence, on the one hand the Multi-step NF-RS achieves \emph{higher} network gains, but on the other hand it leads to \emph{less} fairness  (e.g., 10\% higher CHR and 45\% higher $F_{kl}$ than Greedy among all scenarios). 

We observed the same relation between fairness and network gains, when varying the other system parameters as well; see Table~\ref{tab:fairness-vs-param}. Specifically, increasing the $\alpha$ means that the users choices are affected more by the RS, and the same happens for small $N$ since there are less choices (recommendations); this makes the shaping of the demand caused by a NF-RS more intense, and leads to higher network gains, and as we present in Table~\ref{tab:fairness-vs-param}, less fairness as well. \hmmm{The cache size $C$ does not significantly affect the fairness, but it also had a small effect on the network gains in the scenarios we tested.}

The above observations raise the following question, on which we focus in the next section:
\myquotation{Does great network gain come with great unfairness?}
%
% The above observations indicate that there is an inherent trade-off between the fairness and the network gains achieved \textit{by existing NA-RS}. \pavlos{Here needs better motivation. E.g., having a question (which is intriguing for the reader, and  we answer below) would be nice here.} In the following, we proceed to investigate the behavior and fundamental limits of the network gains vs. fairness trade-offs.

\section{The Fairness vs. Network gains Trade-off}\label{sec:bounds}

% Motivated by the findings of Section~\ref{sec:characterization}, i
In this section we proceed to study the trade-off between the network gains that can be achieved by a NF-RS algorithm and the unfairness it creates. We first analyze the simulation results to verify that such a trade-off exists, and then study it analytically and derive analytic bounds (closed form expressions) for the minimum possible unfairness as a function of the network gains under any NFR scheme. %Finally, we show that frequently the existing NF-RS algorithms operate far from the bounds. 

Let us first formally define the \textit{network gain} $G$ as the increase in the cache hit rate (CHR) achieved by a NF-RS:
\begin{equation}\label{eq:gain-definition}
G = CHR^{NF} - CHR^{BS} = \textstyle %\sum_{i\in\mathcal{C}} p_{i}^{NF} - \sum_{i\in\mathcal{C}} p_{i}^{BS} = 
\sum_{i\in\mathcal{C}} (p_{i}^{NF}- p_{i}^{BS})     
\end{equation}
where $\mathcal{C}\subset\mathcal{K}$ is the set of cached contents. In other words, the network gain is the extra content demand that can be served by the cache when applying a NF-RS.

In Fig.~\ref{fig:fairness-vs-gains} we present scatter plots, where each marker corresponds to a simulation scenario and its ($x,y$)-coordinates correspond to the resulting fairness metric $F$ and network gain $G$ values, respectively. The results verify our previous observations: as the achieved network gain increases, the unfairness of the system increases as well. This positive correlation holds for all fairness metrics. However, the exact behavior differs among the different $F$ metrics (note that all subplots of Fig.~\ref{fig:fairness-vs-gains} present the same simulation scenarios, i.e., with the same network gains); for instance, $F_{max}$ sees a lower increase and with values up to 0.3, while $F_{tv}$ has a larger increase with values up to 1.

In the following theorem we analytically study the observed behavior, and derive theoretical bounds for the %fairness vs. network gain 
trade-off. 
% \begin{mylemma}
% An upper bound for the maximum network gain $G$ that can be achieved by any NA-RS, under a fairness constraint $F\leq c_{F}$, is given by the solution of the optimization problem
% \begin{equation}\label{eq:optimization-problem-bound}
%     \textstyle \max_{p^{NA}} ~~G \hspace{0.5cm} s.t. \hspace{0.2cm} F\leq c_{F}
% \end{equation}
% \end{mylemma}
\begin{theorem}\label{thm:bounds}
Under any NF-RS and any system parameters, for the fairness $F$ vs. network gain $G$ trade-off it holds that
\begin{align}
F_{max} &\geq \frac{1}{C}\cdot G\\
F_{tv}  &\geq G\\
F_{kl}  &\geq -H \cdot \log(1+\frac{G}{H}) - (1-H) \cdot \log(1-\frac{G}{1-H})
% \hmmm{F_{kl}}  & \hmmm{\geq \frac{-H \cdot \log(1+\frac{G}{H}) - (1-H) \cdot \log(1-\frac{G}{1-H})}{\log(\frac{1}{w})}} \nonumber
\end{align}
where $C=|\mathcal{C}|$ is the number of cached contents, and $H=CHR^{BS} = \sum_{i\in \mathcal{C}}p_{i}^{BS}$.
\end{theorem}
\begin{proof}
The proof is given in the Appendix.%~\ref{sec:proof-bounds}.
\end{proof}

% \myitemit{Remark:} While the expression of the $F_{kl}$ metric does not allow us to derive a bound in closed-form following the approach of Theorem~\ref{thm:bounds}, a bound can be numerically calculated as $G\leq f(c_{F})$, $c_{F}\in[0,1]$, by the solution of the following optimization problem
% \begin{equation}\label{eq:optimization-problem-bound}
%     \textstyle \max_{p^{NA}} ~~G \hspace{0.5cm} s.t. \hspace{0.2cm} F_{kl}\leq c_{F}
% \end{equation}

The expressions in Theorem~\ref{thm:bounds} state that the maximum network gain that can be achieved by any NF-RS (i) cannot be larger than the desired $F_{tv}$ value, and (ii) increases with the cache size $C$ in the $F_{max}$ case. This indicates that larger caches can allow network gains without compensating in \textit{individual fairness} ($F_{max}$), while this is not the case in \textit{aggregate fairness} ($F_{tv}$). In the case of $F_{kl}$, the bound is given by a non-linear function (convex in $G$), which depends on the cache size and the distribution of the demand under the BS-RS (captured by the parameter $H=CHR^{BS}$).

Comparing the results in Fig.~\ref{fig:fairness-vs-gains} with the bounds, we can see that in some scenarios the achieved network gains are close to (or, coincide with) the bound\hmmm{, i.e., \textit{the bounds are tight}}\footnote{Note that some of the scenarios/markers in Fig.~\ref{fig:fairness-vs-gain-kl} ($F_{kl}$ case) correspond to different values $CHR^{BS}$ (e.g., due to different $\mathbf{p^{BS}}$ distributions), and thus different bounds. In favor of readability, we avoid depicting several bounds and show only the worst-case bound among those scenarios; i.e., for some scenarios/markers the bound is tighter than the depicted bound.}.

\hmmm{However, in the majority of scenarios,} \textit{the operating point of the considered NF-RS algorithms is far from the bound}\hmmm{. For instance, in %the CDF of the distances between the unfairness of the NF-RS algorithms $F$ and the minimum unfairness given from the bounds $\min F$ (for the same $G$) in 
Fig~\ref{fig:cdf-distance-from-bound} that gives the CDF of the distance (along the $x$-axis) between the operating points and the bound, we can see that in half of the cases (i.e., for 0.5 in the $y$-axis) the resulting unfairness ($x$-axis) is at least $50\%$ larger than the value of the bound for $F_{max}$ and $F_{kl}$ and $20\%$ larger for $F_{tv}$}. The fact that NF-RS algorithms do not operate on the bound can be due to (i) the system parameters (e.g., the QoR constraint) that restrict an algorithm from shaping arbitrarily the content demand, and in this case the bound may not be achievable, or (ii) the NF-RS algorithms themselves, which were designed to optimize the network gain without taking the fairness into account. \hmmm{Thus, a question that follows naturally~is:}
%Using the bounds one can directly evaluate if the operating point of an NA-RS algorithm is already close to the optimal or is amenable to improvements. However, when the operating point is far from the bound, the following question is raised
%
\myquotation{Can a NF-RS be designed to operate closer to the bound, and achieve the optimal fairness vs. network gain trade-off?}
In Section~\ref{sec:design} we address the above question and design an optimal NF-RS algorithm that achieves the maximum network gain under a fairness constraint, and in Section~\ref{sec:price} we study how the introduced constraint affects the network gains.
% \hmmm{We show that under no QoR constraints the bound is always achievable by the Fair NF-RS algorithm, while in Section}

\section{Optimal Fair NFR}\label{sec:design}
% \input{design_v1}

% \theo{IMPORTANT1: sta bounds grafoume sto rhs $F_{fair}$ kai sto opt grafoume $v_{fair}$, prepei na nai coherent auto, tbd with Pavlos}\pavlos{It will be fixed in the optimization section}

% \theo{IMPORTANT2: In this paper we upper bound unfairness, so writing $F(\cdot) \le v_{fair}$ reads weird, because the reader has F for fairness I guess, whereas all the paper discusses unfairness everywhere, if we leave it as is, it should be crystal clear from the very beginning}

In this section we formulate the problem of designing the optimal NF-RS that takes fairness into account. We first model and describe the problem, and then prove that it can be expressed as a linear program (LP) whose solution is the \textit{optimal fair NF-RS}.

% The design of a NF-RS that maximizes the network gain under a fairness constraint, can be expressed as a linear optimization problem (LP):

\myitem{Objective.} The objective in NFR is to maximize the network gains (or, equivalently minimize the network cost), which in our framework is captured by the  CHR, i.e., $\sum_{i\in\mathcal{C}} p_i^{NF}$.

% Our aim is to come up with the recommendation decisions that minimize the expected cost the cache experiences from the user requests. 
% To this end, we assume that the average user, i.e., the one described in Section~\ref{sec:preliminaries}, requests infinitely many items, and so this will approximate the
\myitem{Decision variables.} An NF-RS algorithm selects which contents to recommend, i.e., the recommendation lists $R_{i}$\footnote{There are NF-RS algorithms that select also the network policy, e.g., caching~\cite{sermpezis2018soft,zhu2018coded,chatzieleftheriou2019jointly,costantini2019approximation}. While our framework can be generalized in this direction, this is out of the scope of this paper (see also discussion in Section~\ref{sec:conclusion}).}. We model the recommendation decisions with a set of variables $r_{ij}$, which denote the probability (or frequency) that a content $j$ appears in the recommendation list of $i$, i.e., $r_{ij}=\textnormal{Prob}\{j\in R_{i}\}$. We denote as $\mathbf{R}$ the $K\times K$ matrix that contains all the variables $r_{ij}$ $\forall i,j\in\mathcal{K}$. We follow a probabilistic approach, i.e., $r_{ij}\in[0,1]$ (instead of the deterministic $r_{ij}\in\{0,1\}$), to capture variations of recommendations among different users, or even for the same user (e.g., to not show always the same recommendations for a given content).

\myitem{Constraints.} First, we require that the decision variables are probabilities ($r_{ij}\in[0,1]$) and the recommendation lists contain $N$ recommendations ($\sum_{j\in\mathcal{K}}r_{ij}$=$N$)~\cite{giannakas2020soba,giannakas2019order}. Second, we use a threshold $q\in[0,1]$ for the QoR constraint similarly to previous works, i.e., $\sum_{j\in\mathcal{K}} r_{ij} \cdot u_{ij} \ge q\cdot q^{BS}_{i}$, where $q^{BS}_{i}$ the maximum QoR achieved by the BS-RS. Finally, we introduce the fairness constraint, by using a threshold $c_{f}$ for the maximum allowed unfairness, i.e., $F(\mathbf{p^{BS}}, \mathbf{p^{NF}})\leq c_{f}$, where $F$ is any of the metrics $F_{max}$, $F_{tv}$, or $F_{kl}$. Both thresholds $q$ and $c_{f}$, and the fairness metric $F$, can be selected by the content provider according to its operational requirements. 

In the following theorem we express the above problem as a LP. To do this, we need to introduce a set of auxiliary variables and transform the non-linear expressions in the objective and constraints; the remainder of this section gives the proof, which includes all the needed details.

\begin{theorem}\label{theorem:problem-LP}
The optimal fair NF-RS is given by the solution of the following linear optimization problem: 
% recommendations problem over $\mathbf{z}$, $\mathbf{p}^{NF}$, $K$-sized vectors, and $\mathbf{W}$, a $K\times K$ matrix can be written as
% \begin{problem}[Fair NFR]
\begin{small}
\begin{subequations}\label{problem:NFR-LP}
\begin{align}
\underset{\mathbf{z,p^{NF},W}}{\textnormal{maximize}}~~~& \sum_{i\in\mathcal{C}} p_i^{NF}\label{eq:obj-LP}\\
\textnormal{subject to}~~~& p_j^{NF} - \frac{\alpha}{N} \cdot\sum_{i\in\mathcal{K}} w_{ij} = p^{d}(j), ~\hspace{0.1cm}~\forall j \in \mathcal{K} \label{eq:stationarity-con-LP}\\
& \sum_{j\in\mathcal{K}} w_{ij} \cdot u_{ij} - p_i^{NF} \cdot q\cdot q_i^{BS} \geq 0, ~\hspace{0.1cm}~\forall~i~\in\mathcal{K},
\label{eq:quality-con-LP}\\
& \sum_{j\in\mathcal{K}} w_{ij} - N\cdot p_i^{NF} = 0,~w_{ii} = 0, ~\hspace{0.1cm}~\forall~i\in\mathcal{K} \label{eq:budget-self-con-LP}\\
& w_{ij} - p_i^{NF} \le 0,~w_{ij} \ge 0, ~\hspace{0.1cm}~\forall~i,j\in\mathcal{K} \label{eq:prob-con-lp}\\
& \mathbf{S}(\mathbf{z},\mathbf{p^{NF}})% \le g(v_{fair}, \mathbf{p^{BS}})
\label{eq:fair-con}
\end{align}
\end{subequations}%
\end{small}%
% \pavlos{Add $\mathbf{z}$ in the optimization variables and write \eq{eq:fair-con} as a generic expression that denotes a set of linear constraint, e.g., $\mathbf{S_{z,p^{NF}}}$, so that it is common for all F cases. Then the detailed constraints $\mathbf{S_{z,p^{NF}}}$ for each F-case can be given in a table. E.g., write ``where $\mathbf{S_{z,p^{NF}}}$ is a set of linear constraints given in Table~XX for the different fairness metrics''.}
% \end{problem}
%
where $\mathbf{z}\in \mathbb{R}^{K}$% is a $K$-sized vector
, $\mathbf{W}\in \mathbb{R}^{K\times K}$% a $K\times K$ matrix
, and $\mathbf{S}(\mathbf{z},\mathbf{p^{NF}})$ a set of linear constraints given in Table~\ref{tab:fairness-constraints} for each fairness metric.
% where the function $\mathbf{S}(\cdot)$ represents a set of linear inequalities that describe the fairness constraint. The exact form of $\mathbf{S}(\cdot)$, as well as the need for the auxiliary variable $\mathbf{z}$ will become obvious in the following subsections.
\end{theorem}
\begin{proof}
\begin{table}
\centering
\caption{Set of linear fairness constraints $\mathbf{S}(\mathbf{z},\mathbf{p^{NF}})$.}
\label{tab:fairness-constraints}
\begin{tabular}{ll}
\hline
\hspace{-0.2cm}%
$F_{max}$: &  
\hspace{-0.5cm}%
\begin{tabular}{l}
$p_{i}^{BS} - p_{i}^{NF} \le c_{f}$\\
$p_{i}^{NF} - p_{i}^{BS} \le c_{f}$
\end{tabular}
$~\hspace{0.1cm}~\forall~i\in\mathcal{K}$\\
\hline
\hspace{-0.2cm}%
$F_{tv}$:  & 
\hspace{-0.7cm}%
\begin{tabular}{l}
$\sum_{i\in\mathcal{K}} z_i \le c_{f}$\\
$p_{i}^{BS} - p_{i}^{NF} \le z_i~\hspace{0.1cm}~\forall~i\in\mathcal{K}$\\
$p_{i}^{NF} - p_{i}^{BS} \le z_i~\hspace{0.1cm}~\forall~i\in\mathcal{K}$ 
\end{tabular}\\
\hline
\hspace{-0.2cm}%
$F_{kl}$:  &
\hspace{-0.7cm}%
\begin{tabular}{l}
$\sum_{i\in\mathcal{K}} p^{BS}_i \cdot z_i \ge - \left(c_{f} - \sum_{i\in\mathcal{K}} p^{BS}_i \log(p^{BS}_i)
\right)$\\
$z_i \le e^{(m-1)\cdot s} \cdot p_i^{NF} - (m-1) s - 1,~\hspace{0.1cm} \forall i\in\mathcal{K},m\in\{1,...,M\}$
\end{tabular}%
% \hspace{-0.3cm}%
\\
\hline
\end{tabular}
\end{table}
The objective (and the fairness constraint) involves terms of the content demand $\mathbf{p^{NF}}$, which depends on the recommendations $\mathbf{R}$. In the considered framework% (Sections~\ref{sec:preliminaries-model} and~\ref{sec:sim-setup})
, and similarly to previous works (e.g.,~\cite{chatzieleftheriou2019jointly,giannakas-wowmom-2018,giannakas2019order,zhu2018coded,lin2019content}), the content demand can be modeled with a Markov Chain, with transition probabilities that depend on the recommendations $\mathbf{R}$, the direct requests $\mathbf{p^{(d)}}$ and the probability $\alpha$. Hence, using the result of~\cite{langville2004deeper} (the detailed proof is omitted due to space limitations), we prove the following lemma.% that gives the demand $\mathbf{p^{NF}}$ as a function of $\mathbf{R}$. 
%
%The following Lemma serves as a stepping stone towards the proof of Theorem~\ref{theorem:problem-LP}.
\begin{mylemma}\label{thm:total-demand} The content demand $\mathbf{p^{NF}}$ is given by 
\begin{equation}\label{eq:lemma-total-demand}
\mathbf{p^{NF}} = (1-\alpha) \cdot \mathbf{p^{(d)}} \cdot \left(\mathbf{I}-\frac{\alpha}{N} \cdot \mathbf{R} \right)^{-1}
\end{equation}
for $\alpha \in (0,1)$ and $\mathbf{p^{(d)}}>0$;
$\mathbf{I}$ is the $K\times K$ identity matrix.% (of size $K\times K$). %This expression for the content demand is true iff $\alpha \in (0,1)$ and $\mathbf{p_{(d)}}>0$ elementwise \cite{langville2004deeper}.
\end{mylemma}
% removed a remark written by Pavlos
% \noindent \textit{Remark:} The cases of $\alpha$=\{0,1\} or $p_{i}^{(d)}$=0 can be trivially solved, and thus we do not consider them in the remainder. 

% \begin{corollary}\label{thm:p-is-positive}
% If we can write the content demand as in Lemma~\ref{thm:total-demand}, $p_i > 0~\forall~i\in\mathcal{K}$.
% \end{corollary}
% %
% \begin{proof}
% Multiplying both sides of the expression in Lemma~\ref{thm:total-demand} with $\left(\mathbf{I}-\alpha \cdot \mathbf{P_{(r)}} \right)^{-1}$ yields
% \begin{equation}\label{eq:stationarity}
%     \mathbf{p^{NA(\textit{T})}} = \alpha\cdot \mathbf{p^{NA(\textit{T})}} \mathbf{P_{r}} + (1-\alpha) \cdot \mathbf{p_{(d)}}^{T}
% \end{equation}
% We can observe that $\mathbf{p} = f(\mathbf{p}, \mathbf{P_{(r)}}) + (1-\alpha) \cdot \mathbf{p_{d}}$, where $f(\mathbf{p}, \mathbf{P_{(r)}})\ge0$ as a dot product of nonzero vectors and  $(1-\alpha) \cdot \mathbf{p_{d}}>0$ for $\alpha\in (0,1)$ and $\mathbf{p_d}>0$, which proves the claim.
% % On the LHS we have the desired quantity $\mathbf{p}$, which some function of itself plus some
% % The condition $p_0(j)>0~\forall i\in\mathcal{K}$ and $f_{ij}\ge0$, imply that $p_j > (1-\alpha)p_0(j)>0$, which ensures that the denominator in $r_{ij} = \frac{f_{ij}}{p_i}$ is nonzero, thus making it always possible to go from $\{f_{ij},p_i\} \to r_{ij}$.
% \end{proof}

% The optimization problem considering \emph{only the quality, and without fairness control} has been recently studied and in its original form was shown to be nonconvex in~\cite{giannakas-wowmom-2018}.
% \begin{proof}
Lemma~\ref{thm:total-demand} gives $\mathbf{p}^{NF}$ as a function of an inverse matrix of $\mathbf{R}$, which in general is non-convex on the variable $\mathbf{R}$. To overcome this non-linearity, we explicitly introduce $\mathbf{p^{NF}}$ as an auxiliary optimization variable. The only constraint we need for the variable $\mathbf{p^{NF}}$ is \eq{eq:lemma-total-demand}. To express this constraint as a linear equation, we first multiply both sides with the term $\left(\mathbf{I}-\frac{\alpha}{N} \cdot \mathbf{R} \right)$ and write:
\begin{equation}\label{eq:stationarity}
    \mathbf{p^{NF}} - \frac{\alpha}{N}\cdot \mathbf{p^{NF}} \cdot\mathbf{R} = (1-\alpha) \cdot \mathbf{p^{(d)}}
\end{equation}
Since \eq{eq:stationarity} involves products of the variables $\mathbf{p^{NF}} \cdot\mathbf{R}$ (i.e., a quadratic term), we substitute the optimization variables $r_{ij}$ with the new auxiliary variables $w_{ij}$, where $r_{ij}=\frac{w_{ij}}{p_{i}^{NF}}$. This substitution is possible because $p_{i}^{NF}>0$ for the cases of interest, as stated in the following corollary (whose proof follows by observing \eq{eq:stationarity}) .
\begin{corollary}\label{thm:p-is-positive}
$p_{i}^{NF} > 0$, $\forall i\in\mathcal{K}$, for $\alpha \in (0,1)$ and $\mathbf{p^{(d)}}>0$.
\end{corollary}

Having introduced the new auxiliary variables, it is easy to show how the constraints of \eq{problem:NFR-LP} are derived, by substituting $r_{ij} = \frac{w_{ij}}{p_i^{NF}}$, as follows:
%
% Our objective is the maximization of the expected cache hit rate as in Eq. (\ref{eq:obj-LP}). To see why the constraints of Eq. (\ref{problem:NFR-LP}) solve the recommendation problem, we set $r_{ij} = \frac{w_{ij}}{p_i^{NF}}$. Due to Corollary~\ref{thm:p-is-positive}, this is possible as $p_i^{NF}$ is strictly positive. We have
\begin{subequations}
\begin{align}
& \textnormal{Eq. (\ref{eq:stationarity-con-LP})} \Leftrightarrow p_{j}^{NF} - \frac{\alpha}{N}\cdot\sum_{i\in\mathcal{K}} p_i^{NF}\cdot r_{ij} = (1-\alpha)\cdot p^{d}_{i}\label{eq:stationarity-con-nfr}\\
& \textnormal{Eq. (\ref{eq:quality-con-LP})} \Leftrightarrow  \sum_{j\in\mathcal{K}} r_{ij} \cdot u_{ij} \ge q\cdot q^{BS}_{i}%, %~\hspace{0.1cm}~\forall~i\in\mathcal{K} 
\label{eq:quality-con-nfr}\\
& \textnormal{Eq. (\ref{eq:budget-self-con-LP})} \Leftrightarrow \sum_{j\in\mathcal{K}} r_{ij} = N,~r_{ii} = 0%,~\hspace{0.1cm}~\forall~i\in\mathcal{K}
\label{eq:budget-self-con-nfr}\\
& \textnormal{Eq. (\ref{eq:prob-con-lp})} \Leftrightarrow r_{ij} \le 1,~r_{ij} \ge 0%,~\hspace{0.1cm}~\forall~i,~j\in\mathcal{K} 
\label{eq:prob-con-nfr}
\end{align}
\end{subequations}
where \eq{eq:stationarity-con-nfr} is equivalent to \eq{eq:stationarity} (and guarantees that $\mathbf{p^{NF}}$ is a stationary distribution for $\mathbf{R}$), \eq{eq:quality-con-nfr} is the QoR constraint, and \eq{eq:budget-self-con-nfr} and \eq{eq:prob-con-nfr} are constraints on the recommendation variables.
% where Eq. (\ref{eq:stationarity-con-nfr}) guarantees that $\mathbf{p^{NF}}$ is a stationary distribution for $\mathbf{P_{(r)}}$, Eq. (\ref{eq:quality-con-nfr}) controls the quality threshold, and Eq. (\ref{eq:budget-self-con-nfr}-{eq:prob-con-nfr}) are the budget and probability constraints.
% \end{proof}

Up to now, we have transformed all the constraints, apart from the fairness constraint $F(\mathbf{p^{BS}}, \mathbf{p^{NF}})\leq c_{f}$. In the following, we transform the fairness constraint in a set of linear constraints $\mathbf{S}(\mathbf{z},\mathbf{p^{NF}})$ for each fairness metric of Section~\ref{sec:fairness-definition}.

\myitem{F-max.} %\myitem{Fairness constraint for the $F_{max}$ case.}%\subsection{Max fairness}
In the case of $F_{max}$ the fairness constraint is % becomes We begin by considering the max fairness metric discussed in Section~\ref{sec:preliminaries}, and constrain the feasible solutions of Eq. (\ref{problem:NFR-LP}) by enforcing
\begin{align}\label{eq:max-con}
    F_{max}(\mathbf{p^{NF}}, \mathbf{p^{BS}}) = \max_{i\in\mathcal{K}} \{|p_{i}^{NF}-p_{i}^{BS}|\} \le c_{f}
\end{align}
%with $c_{f}\in[0,1]$. 
\eq{eq:max-con} is not a linear inequality. However, it can be expressed as the intersection of the following $2\cdot K$ linear inequalities
\begin{equation}\label{eq:max-linear-con}
\begin{tabular}{c}
$p_{i}^{BS} - p_{i}^{NF} \le c_{f}$\\
$p_{i}^{NF} - p_{i}^{BS} \le c_{f}$
\end{tabular}
    % p_{i}^{BS} - p_{i}^{NF} \le c_{f},~\hspace{0.1cm}~p_{i}^{NF} - p_{i}^{BS} \le c_{f},~\hspace{0.1cm}~\forall~i\in\mathcal{K}
~\hspace{0.1cm}~\forall~i\in\mathcal{K}
\end{equation}
where we first set $|p_{i}^{NF} - p_{i}^{BS}| \le c_{f}$ $\forall i\in\mathcal{K}$ as equivalent to constraining the $\max$, and then substituted each absolute term $|x|\leq c_{f}$ with two constraints $x\leq c_{f}$ and $-x\leq c_{f}$.
% Therefore, we have expressed Eq.(\ref{eq:max-con}) as $2 K$ extra \emph{linear} inequalities, which means we can solve max-fair NFR optimally.

\myitem{F-tv.} %\myitem{Fairness constraint for the $F_{tv}$ case.}%\subsection{TV fairness}
A similar approach could be applied for the constraint
\begin{equation}\label{eq:tv-con}
\textstyle F_{tv}(\mathbf{p^{NF}}, \mathbf{p^{BS}}) = \frac{1}{2}\cdot \sum_{i\in\mathcal{K}}|p_{i}^{NF}-p_{i}^{BS}| \leq c_{f}   
\end{equation}
However, it would lead to $2^K$ linear inequalities, which is impractical for large catalogs. Hence, we introduce an auxiliary set of variables $\mathbf{z}\in\mathbb{R}^{K}$ (a $K$-sized vector) and substitute \eq{eq:tv-con} with the following constraints
% We now discuss how to transform the constraint $F_{TV}(p,p^{BS}) \le v_{TV}$. 
% %With the current variable description, one would need  in order to fully encode the sum of absolute values, 
% If we attempt to list every linear inequality conditioned on all possible $\{+,-\}$ signs of the summands, we conclude that $2^K$ linear inequalities are needed; an approach which is of course impractical for large catalogs. %
% However, adding an auxiliary set of variables $\mathbf{z}$ (a $K$-sized vector) resolves this issue as it allows us to do the following,
\begin{align}\label{eq:tv-transformed}
\begin{tabular}{l}
$\sum_{i\in\mathcal{K}} z_i \le c_{f}$\\
$|p^{BS}_{i} - p_{i}^{NF}| \le z_i,  ~\hspace{0.1cm}~\forall~i\in\mathcal{K}$
\end{tabular}
\end{align}
The first constraint is a linear inequality, and the remaining $K$ inequalities of \eq{eq:tv-transformed} can be substituted with $2\cdot K$ linear inequalities similarly to \eq{eq:max-linear-con}.% new inequalities that handle the sign of the absolute value, two for each $i\in\mathcal{K}$. It is hence possible to solve the $TV$-fair NFR problem optimally facing it as an LP, but now over three set of variables, i.e., $\mathbf{p}^{NF}, \mathbf{z, W}$. Note that in $TV$ and $\max$ fairness cases, we need to add a number of inequalities \emph{which is linear} on $K$.

\myitem{F-kl.} %\myitem{Fairness constraint for the $F_{kl}$ case.}%\subsection{KL fairness}
In the $F_{kl}$ case, the constraint can be written as %Here we discuss the KL-divergence case, which is the last of the three fairness metrics we examined in Section~\ref{sec:preliminaries}.
\begin{align}\label{eq:kl-basis}
\textstyle  F_{kl}  %= \sum_{i\in\mathcal{K}} p^{BS}_i \log\frac{p^{BS}_i}{p_i^{NF}} 
= \sum_{i\in\mathcal{K}} p^{BS}_i \cdot \left(\log(p^{BS}_i) - \log(p_i^{NF})\right) \le c_{f}
\end{align}
\eq{eq:kl-basis} involves a logarithmic function, thus, we cannot proceed as in $F_{max}$ or $F_{tv}$.
%The fundamental difference with the rest of the fairness metrics, is that it is inherently nonlinear.
% \theo{need to re-read carefully}
% Furthermore, note that explicitly lower-bounding the $F_{KL} \ge 0$ is not needed, since the variable $p$ is a probability vector, and thus by definition $KL(q,p) > 0$ unless $q=p$ with $KL(q,p) = 0$.
We first rewrite \eq{eq:kl-basis} as %However, the set defined by the inequality $\sum_{i} -\log(x_i) \le w$ for $x_i>0$, as in Eq.(\ref{eq:kl-basis}), is convex. 
%
% In order to keep a more compact framework, we choose to use a linearized version of the $\log$ function and approximate it with a general family of linear cuts. 
% This process can be split in two steps.
\begin{align}\label{eq:f-kl-constraint-intermediate}
\textstyle \sum_{i\in\mathcal{K}} p^{BS}_i \log(p_i^{NF}) \ge - \left(c_{f} - \sum_{i\in\mathcal{K}} p^{BS}_i \log(p^{BS}_i)%\underbrace{c_{f} - \sum_{i=1}^K p^{BS}_i \log(p^{BS}_i}_{B})
\right)
\end{align}
where we remind that $p_{i}^{NF}$ are optimization variables and $p_{i}^{BS}$ are given constants. % since $p^{BS}_i \le 1$ as a probability, implies that $B$ is strictly positive.
%$v_{KL} - \sum_{i=1}^K p^{BS}_i (\log(p^{BS}_i) > 0$. 
% Therefore we demand the lhs which is a negative quantity to be higher than another negative quantity, i.e., the one on the rhs.
Then, we introduce an auxiliary set of $K$ variables $\mathbf{z}\in\mathbb{R}^{K}$, and we demand the following $K+1$ inequalities which are equivalent to \eq{eq:f-kl-constraint-intermediate}
\begin{align}\label{eq:kl-transformed}
\begin{tabular}{l}
$\sum_{i\in\mathcal{K}} p^{BS}_i \cdot z_i \ge - \left(c_{f} - \sum_{i\in\mathcal{K}} p^{BS}_i \log(p^{BS}_i)
\right)$\\
$\log(p_i^{NF}) \ge z_i, ~\hspace{0.1cm}~\forall~i\in\mathcal{K}$
\end{tabular}
%   \sum_{i\in\mathcal{K}} p^{BS}_i z_i \ge -B,~\log(p_i) \ge z_i \forall~i\in\mathcal{K}.
\end{align}
% Therefore, in order to have a feasible solution, it suffices to demand
% \begin{align}\label{eq:kl-z}
%   \sum_{i=1}^K p^{BS}_i z_i \ge -B 
% % \bigg( v_{KL} - \sum_{i=1}^K p^{BS}_i (\log(p^{BS}_i)\bigg)
% \end{align}
The first inequality of \eq{eq:kl-transformed} is linear. The remaining $K$ inequalities are nonlinear due to the presence of the logarithm. To transform them to linear constraints, we approximate the logarithm with a general family of linear cuts. Specifically, we define $M$ lines for every $i$, as $f(p_i^{NF}) = a_{m,i}\cdot p_i^{NF} + b_{m,i}$, that are tangent to the $\log(p_i^{NF})$ function in the interval $p_i^{NF}\in[0,1]$. Essentially we sample the logarithm at the points %\theo{prosoxi edw, need to recheck}
% \[\{p_i^{NF}~,~\log(p_i^{NF})\} =
\[\{e^{-(m-1)\cdot s}~,~\log e^{-(m-1)\cdot s}\}\]
%$x_m = e^{-(m-1)\cdot s}$, 
%
where $m = 1, \dots, M$ and $s<1$. The $M$ slopes $a_{m,i}$ and the corresponding constants $b_{m,i}$, which are the same for every dimension $i\in\mathcal{K}$, of these tangent lines can be straightforwardly calculated.
%
% We know that $a_m = e^{(m-1)\cdot s}$, since the slope of the lines is exactly the slope of the $\log$ function at these points. In order to calculate $b_m$, i.e., where the line crosses the $y$-axis, we use that the logarithm and the line when evaluated at point $x$, have the exact same value, as $\log(x) = a_m x + b_m$, and solve for $b_m$. 
%
Thus, instead of using the $K$ non-linear inequalities of \eq{eq:kl-transformed}, we use the following $M\cdot K$ inequalities that are linear on the variables $z_{i}$ and $p_{i}^{NF}$
\begin{align*}
\textstyle z_i \le e^{(m-1)\cdot s} \cdot p_i^{NF} - (m-1) s - 1,~\hspace{0.1cm} \forall i\in\mathcal{K},m\in\{1,...,M\}
\end{align*}
\textit{Remark:} The sampling step $s$ and the number of linear cuts $M$, play a major role in the optimization process. $s$ that is small enough for dense sampling% of the interval $[0,1]$
, and an $M$ according to the size of the catalog. In our scenarios, we found that $s=0.05$ and $M=160$ suffices for a catalog of $K\approx 1000$ contents.
% \theo{actually this algorithm could be infeasible, but not unboundedly as we thought, I made a short analysis for that.}
\end{proof}

\section{The Price of Fairness}\label{sec:price}
In this section, we employ the Fair NF-RS of Section~\ref{sec:design} to the simulation setup of Section~\ref{sec:sim-setup}. We consider different values for the fairness constraints $c_{f}$, and in Fig.~\ref{fig:price} we present the performance, i.e., the achieved CHR ($y$-axis) of the Fair NF-RS (continuous lines) vs. the resulting unfairness ($x$-axis). We present two indicative scenarios for the LastFM/Movielens datasets (red/blue color). Also, we present the bound for each scenario (dash lines), the operating points \textit{\{unfairness, CHR\}} of the BS-RS (star markers) and the NF-RS schemes that do not consider fairness (star, cross, and hexagonal markers for the Greedy, Multi-step, and CABaRet NF-RS, respectively). 

Below we discuss the main findings stemming from Fig.~\ref{fig:price}, which provide useful insights for the effect of imposing fairness in NFR and the price we have to pay for this.

\begin{figure*}
\centering
\subfigure[{$F_{max}$}]{\includegraphics[width=0.49\columnwidth]{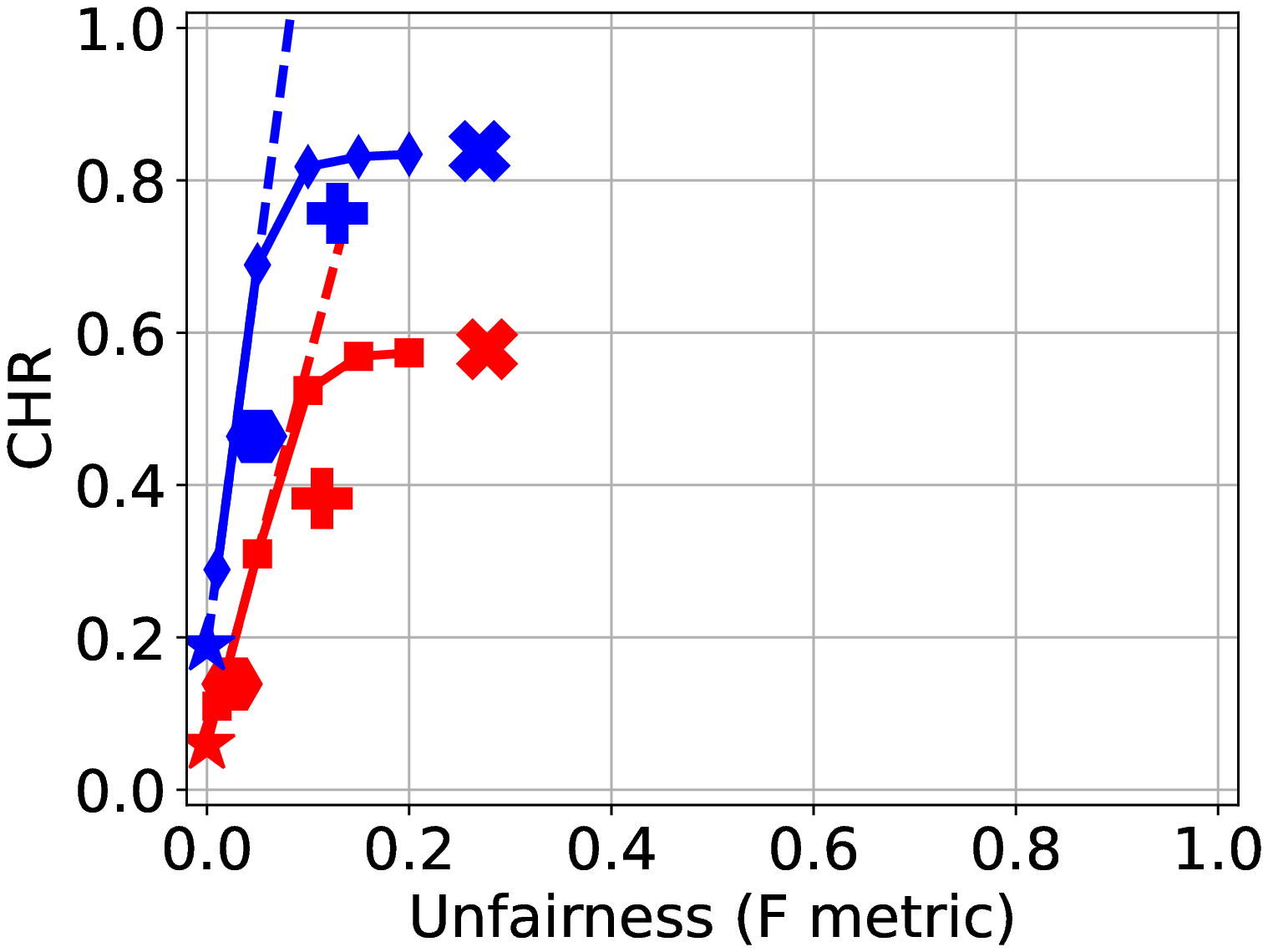}\label{fig:price-max}}
\subfigure[{$F_{tv}$}]{\includegraphics[width=0.49\columnwidth]{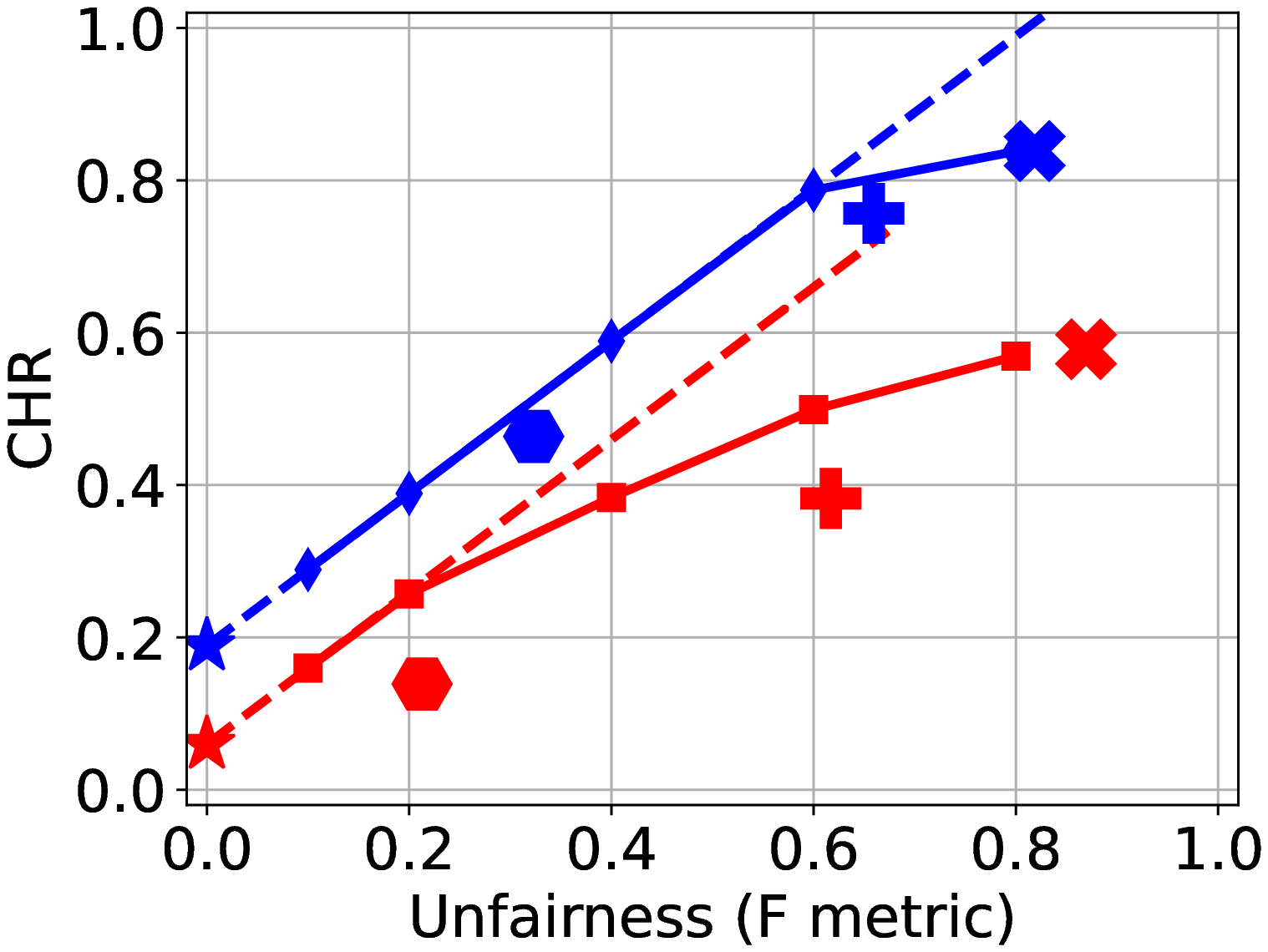}\label{fig:price-avg}}
\subfigure[{$F_{kl}$}]{\includegraphics[width=0.49\columnwidth]{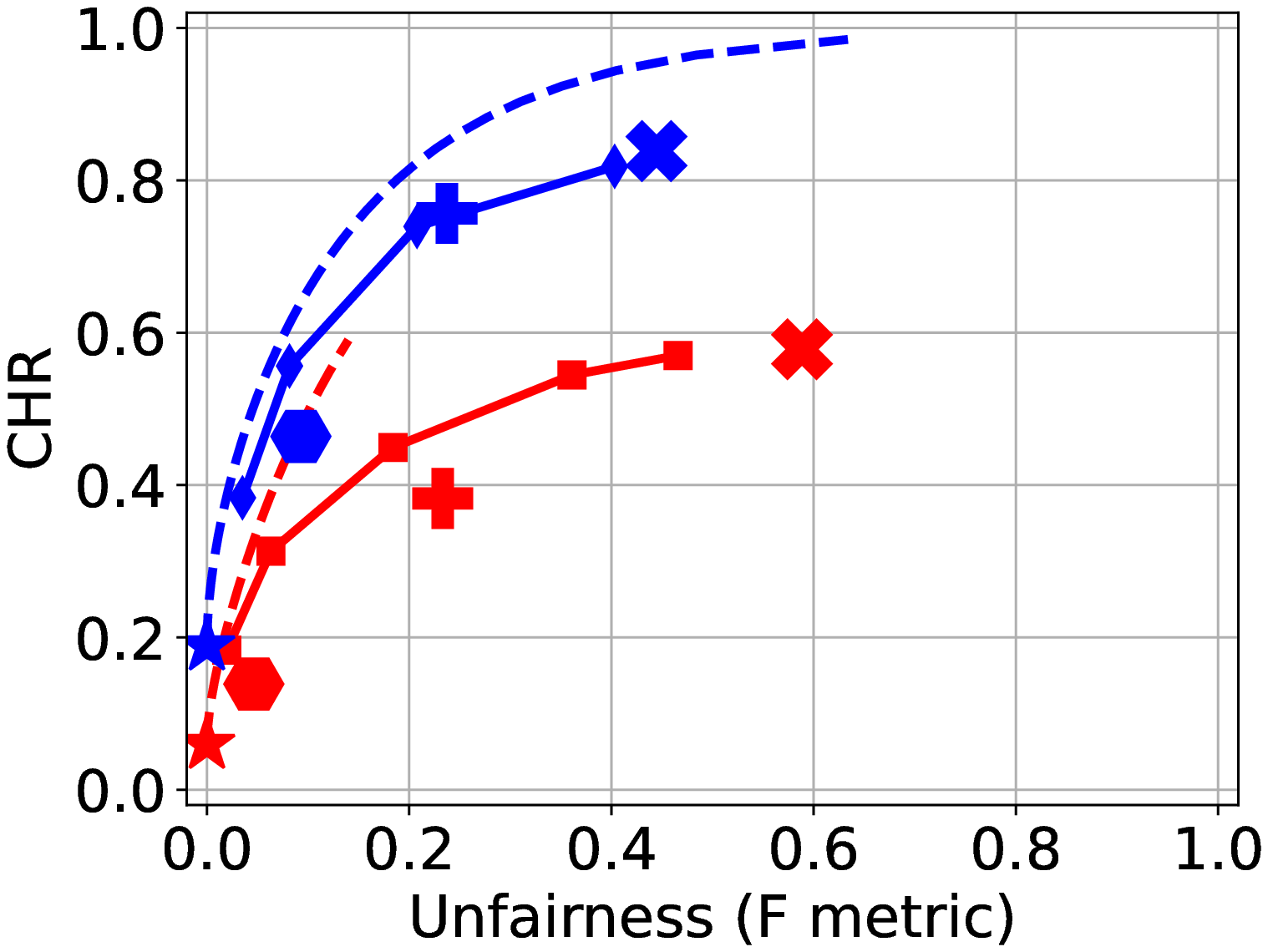}\label{fig:price-kl}}
\subfigure{\includegraphics[width=0.49\columnwidth]{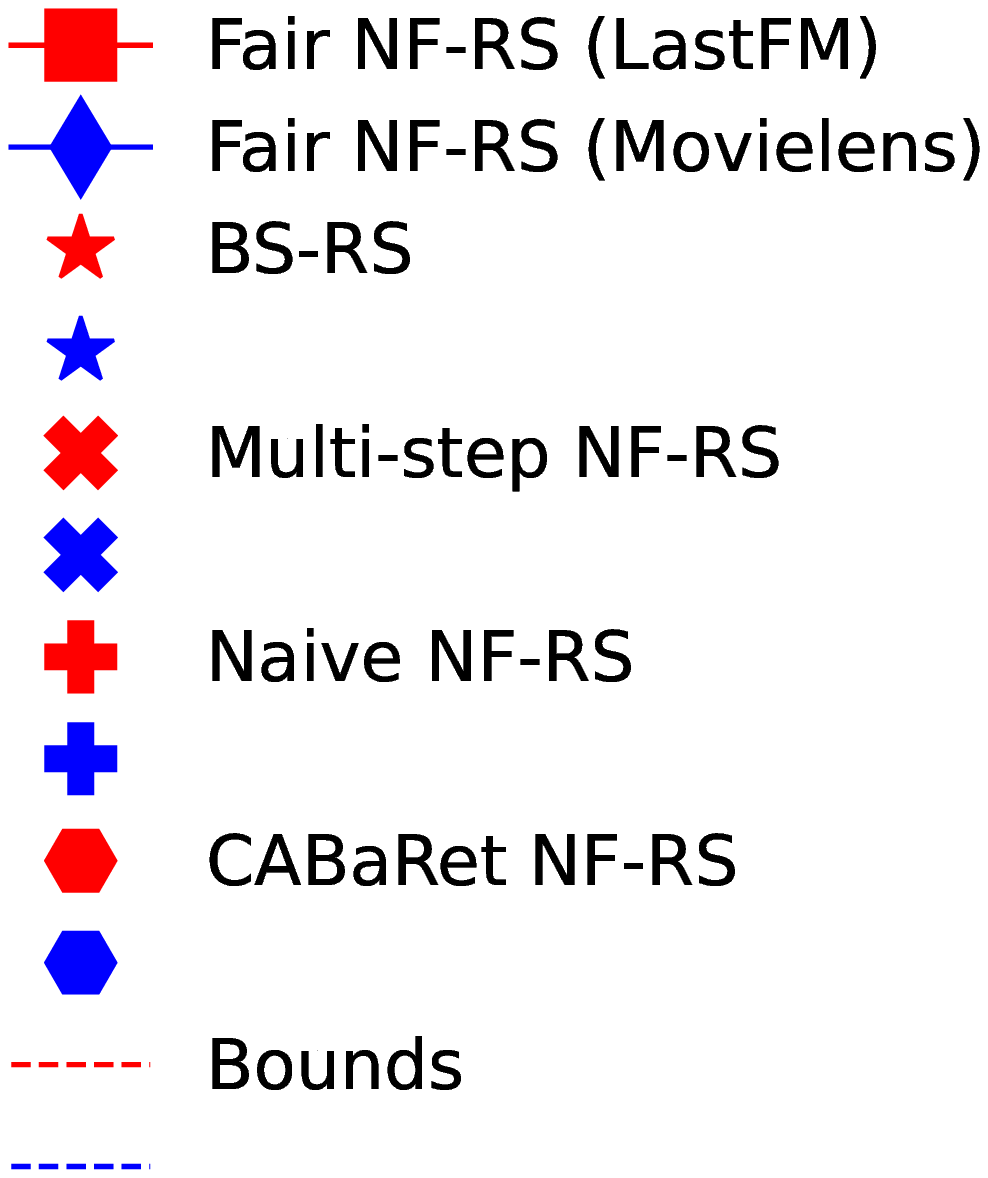}}
\caption{The price of fairness: Comparison of the performance (fairness at $x$-axis vs. CHR at $y$-axis) of the Fair NF-RS (continuous lines) with other RS (markers) and the bounds (dashed lines). Red colors correspond to the LastFM dataset scenario and blue colors to the Movielens dataset scenario, with parameters $\alpha$=0.99, $N$=2, $q$=0.9, $C$=5 (LastFM) and $C$=10 (Movielens).}
\label{fig:price}
\end{figure*}

\keyfinding{1}{The Fair NF-RS always achieves a better performance trade-off than other NF-RS algorithms}

The first observation that verifies the correctness of the proposed approach is that the Fair NF-RS performs better than other NF-RS (both in fairness and network gains), i.e., the curve of the Fair NF-RS is above (higher CHR) and/or on the left (less unfairness)  of the markers that indicate the operation points of the other NF-RS algorithms. Extending the Fair NF-RS curve towards (i) small values of $F$ ($x$-axis) leads to the operating point of the BS-RS ($F=0$), and (ii) large values of $F$ leads to the operating point of the Multi-step NF-RS, which is equivalent to the Fair NF-RS \emph{without} fairness constraint. 

\keyfinding{2}{By allowing a little unfairness, high network gains can be achieved}

Comparing the Fair NF-RS performance with this of the BS-RS, we see that the increase in the network gains is steep for a small relaxation in the unfairness (i.e., for small values in the $x$-axis). In fact, we can see that the curve of the Fair NF-RS coincides with the bound, which means that the optimal fairness-gains trade-off is achievable by the Fair NF-RS for small values of fairness constraints. \textit{This is a promising message for the practical feasibility of the NFR framework}: significant network gains are possible even when a level of fairness is required by the content provider.

\keyfinding{3}{The price (wrt. the network gain) of imposing fairness is small}

Moving our attention on the other side of the Fair NF-RS curve (i.e., for higher values of $F$), two interesting observations can be made: (i) the curve of the Fair NF-RS is concave, and (ii) the gains in CHR diminish for large values of $F$. These findings show that \textit{similar network gains} to the Multi-step NF-RS (which achieves the best performance under no fairness constraints) can be achieved \textit{with much less unfairness}. In particular in the case of $F_{max}$ that captures the individual fairness, this behavior is clearer: a CHR very close to the highest possible can be achieved by the Fair NF-RS even with 3 times lower $F_{max}$ compared to the Multi-step NF-RS .

Moreover, we can see that while the bounds are linear in the case of $F_{max}$ and $F_{tv}$, the curve of the Fair NF-RS is concave: this is a positive finding indicating that the Fair NF-RS stays close to the bound and deviates for it only when large values of unfairness are allowed. Even in the case of $F_{kl}$ where the bound is also concave, the Fair NF-RS curve approaches the highest possible CHR with a faster rate.

\section{Related Work}\label{sec:related}
\myitem{NFR.} The paradigm of network-friendly (or, network-aware) recommendations has been recently proposed and studied under different network setups and content services~\cite{sch-chants-2016,chatzieleftheriou2017caching,sermpezis2018soft,giannakas-wowmom-2018,kastanakis-cabaret-mecomm-2018,zhu2018coded,chatzieleftheriou2019jointly,costantini2019approximation,garetto2020similarity,qi2018optimizing,giannakas2019order,chatzieleftheriou2019joint,gupta2019effect,lin2018joint,song2018making,lin2019content,sermpezis2019towards,cache-centric-video-recommendation,content-recommendation-swarming}. The proposed NFR schemes aim to increase the network gains (and/or improve the quality of content delivery) by selecting recommendations~\cite{cache-centric-video-recommendation,giannakas-wowmom-2018,giannakas2019order,kastanakis-cabaret-mecomm-2018} or by jointly designing the recommendation and network policy (e.g., caching)\cite{chatzieleftheriou2017caching,sermpezis2018soft,zhu2018coded,chatzieleftheriou2019jointly,costantini2019approximation,garetto2020similarity,qi2018optimizing,chatzieleftheriou2019joint,gupta2019effect,lin2018joint,song2018making,lin2019content}. The majority of related works considers \textit{cache-friendly} recommendations in mobile networks~\cite{sch-chants-2016,chatzieleftheriou2017caching,sermpezis2018soft,giannakas-wowmom-2018,zhu2018coded,chatzieleftheriou2019jointly,costantini2019approximation}. However, the same principles apply to generic network setups~\cite{giannakas2019order}, such as coded caching~\cite{zhu2018coded}, broadcast communications~\cite{lin2018joint,song2018making,lin2019content}, user association to base stations~\cite{chatzieleftheriou2019joint}, or swarming systems~\cite{content-recommendation-swarming}. 
%
% More specifically, in~\cite{cache-centric-video-recommendation}, the authors considered only the the Independent Request Model (IRM), while more recent works like~\cite{giannakas-wowmom-2018,giannakas2019order} attempt to solve the same problem in the stationary regime, i.e., taking into account a long sequence of requests. In~\cite{song2018making}, the effect of recommendations is investigated but under the lens of information theory. The second problem found in the literature is the co-design of caching and recommendation. Along this direction, there is a variety of works~\cite{kastanakis-cabaret-mecomm-2018,zhu2018coded,lin2018joint,qi2018optimizing,chatzieleftheriou2019joint,gupta2019effect,lin2019content,sermpezis2019towards,costantini2019approximation}. Due to the fundamental hardness of the problem, the problem is decomposed into (a) caching and then (b) recommendation decisions with an ultimate goal to maximize the hit rate.
%
While some of the proposed schemes take into account the user perspective by accounting the QoR, none of them has considered the fairness in recommendations from the perspective of the content provider. In this context, our work studies the dimension of fairness, thus providing a more complete view of the NFR paradigm. The proposed Fair NF-RS retains the efficiency of previous NFR schemes for achieving high network gains, while reduces the unfairness.

\myitem{Fairness in RS.}
A variety of fairness metrics are used by the RS community~\cite{abdollahpouri2019multi,burke2017multisided,burke2018balanced,edizel2020fairecsys,patro2020incremental,pessach2020algorithmic,Sacharidis2019ACA,steck2018calibrated,yang2017measuring,liu2018personalizing} to capture different notions and needs of the content providers.
Moreover, the fairness in RS can be defined with respect to the consumer (c-fairness) or the provider (p-fairness)~\cite{abdollahpouri2019multi,burke2017multisided}. The former is typically used to design recommendation algorithms whose output is independent of sensitive user traits, e.g., race or gender~\cite{burke2017multisided,edizel2020fairecsys}. In other words, c-fairness aims to capture discrimination between users. Hence, it is orthogonal to our study, e.g., it could be considered as a part of the BS-RS
%, with which we already compare the NF-RS,
and depends on the recommendation scores $u_{ij}$ for which we consider a generic definition. The notion of p-fairness, which we use in this paper, aims to capture potential discrimination of the content provider towards different content producers/owners (or, individual contents). The proposed Fair NF-RS provides recommendations that achieve a balance between the user satisfaction (QoR), the content provider (fairness), and the network gains. A similar issue is addressed in~\cite{liu2018personalizing}, from a multi-stakeholders perspective.

\section{Conclusion}\label{sec:conclusion}
Previous works have shown that NFR can bring significant gains for the network, however, without considering the fairness, which is a key factor for content providers. This work is the first to study the dimension of fairness in NFR, and explore the trade-offs between controlling fairness and increasing network gains. %Using real traces, we show that recently developed NF-RS tend to concentrate the content demand in specific subset of contents and thus significantly change the projected plans of the CPs.
Our results show that fairness \textit{need} and \textit{can} be taken into account in NFR, while the price (wrt. network cost) that one has to pay to impose fairness is small.

% Our findings show that that fairness needs to be established as a hard system specification when designing NF-RS.
%
% To that end, we propose a simple optimization framework based on Linear Programming and solve the joint problem of NF-RS under fairness constraints. Our results show that for some fairness measures, it is possible to sacrifice minimal hit rate in order to achieve high fairness and in other cases that the price of fairness is extremely high. 
%
% Finally, an interesting extension would be to see how joint opt caching and recommendation. Intuitively joint approaches

We believe that the findings of this paper can motivate further research on fairness in NFR. For example, under NFR schemes that \textit{jointly} select the recommendation and network policies, we expect a more aggressive shaping of the demand. Hence, it is of interest to investigate if, and how the introduced unfairness and the trade-offs change under such schemes. In terms of fairness notions, an extension can be towards group fairness~\cite{Sacharidis2019ACA,xiao2017fairness}, where the contents belong to classes (e.g., of the same genre or producer)~\cite{liu2018personalizing,mehrotra2018towards}, and the fairness is defined among the aggregate demand of content classes. This more relaxed fairness metric, probably allows more flexibility in the decisions of the NF-RS and, thus, higher network gains. 

\section*{Appendix: Proof of Theorem~\ref{thm:bounds}}

% \section{Proof of Theorem~\ref{thm:bounds}}
% \label{sec:proof-bounds}

The bound for the network gain $G$ vs. fairness $F$ trade-off is given by the solution of the optimization problem%\footnote{Equivalently, it is given by the problem: $\max_{\mathbf{p^{NF}}} F$, s.t. $G\leq c_{g}$.}
\begin{footnotesize}
\begin{equation}\label{eq:optimization-problem-bound}
    \textstyle\max_{\mathbf{p^{NF}}} ~~G \hspace{0.5cm} s.t. \hspace{0.2cm} F\leq c_{F}
\end{equation}%
\end{footnotesize}%
where $G$ is defined in \eq{eq:gain-definition}, $F$ is the fairness metrics (such as the $F_{max}, F_{tv}, F_{kl}$ of Section~\ref{sec:fairness-definition}), $c_{F}$ a constant, and $\mathbf{p^{NF}}$ has to be a probability distribution.

\myitem{F-max.} For the network gain it holds that
\begin{footnotesize}
\begin{align*}
G   \leq \sum_{i\in\mathcal{C}} |p_{i}^{NF}- p_{i}^{BS}| \leq C\cdot \max_{j\in\mathcal{K}}|p_{j}^{NF}- p_{j}^{BS}| = C\cdot F_{max}
\end{align*}
\end{footnotesize}%
where the first inequality follows by applying the property $x\leq |x|$ to the terms in the expression of $G$ (\eq{eq:gain-definition}), the second inequality holds because $|p_{i}^{NF}- p_{i}^{BS}|\leq \max_{j\in\mathcal{K}}|p_{j}^{NF}- p_{j}^{BS}|$, $\forall i\in \mathcal{C}$, since $\mathcal{C} \subset \mathcal{K}$, and in the last equality we simply substituted from the definition of $F_{max}$ (Section~\ref{sec:fairness-definition}). 

\myitem{F-tv.} Starting similarly to the $F_{max}$ case, we get
\begin{footnotesize}
\begin{align}
G   \leq \sum_{i\in\mathcal{C}} |p_{i}^{NF}- p_{i}^{BS}| = 2\cdot F_{tv} - \sum_{i\in\mathcal{K}\backslash\mathcal{C}} |p_{i}^{NF}- p_{i}^{BS}| \label{eq:derivation-bound-avg-step-1}
\end{align}
\end{footnotesize}%
where the equality follows from the definition of $F_{tv}$. % (Section~\ref{sec:fairness-definition}).
The right hand side of \eq{eq:derivation-bound-avg-step-1} increases when $\sum_{i\in\mathcal{K}\backslash\mathcal{C}} |p_{i}^{NF}- p_{i}^{BS}|$ decreases. The min value of this term can be calculated as:

% \vspace{-\baselineskip}
% \begin{footnotesize}
% \begin{align}
% \sum_{i\in\mathcal{K}} p_{i}^{BS}     &= \sum_{i\in\mathcal{K}} p_{i}^{NF}  \Rightarrow \nonumber\\
% \sum_{i\in\mathcal{K}\backslash\mathcal{C}} (p_{i}^{BS}-p_{i}^{NF})     &=\sum_{i\in\mathcal{C}} (p_{i}^{NF}-p_{i}^{BS})  \Rightarrow\nonumber\\
% % \sum_{i\in\mathcal{K}\backslash\mathcal{C}} |p_{i}^{BS}-p_{i}^{NF}|     &\geq \sum_{i\in\mathcal{C}} (p_{i}^{NF}-p_{i}^{BS}) \Rightarrow\nonumber\\
% \sum_{i\in\mathcal{K}\backslash\mathcal{C}} |p_{i}^{BS}-p_{i}^{NF}|     &\geq G \label{eq:derivation-bound-avg-step-2}
% \end{align}
% \end{footnotesize}
%%%%%%%%%%%%%%%%%%
%%%%%%%%%%%%%%%%%%
%%%%%%%%%%%%%%%%%%
%%%%%%%%%%%%%%%%%%
\vspace{-\baselineskip}
\begin{footnotesize}
\begin{align}
&\sum_{i\in\mathcal{K}} p_{i}^{BS} = \sum_{i\in\mathcal{K}} p_{i}^{NF} ~~~\Rightarrow~~~ 
\sum_{i\in\mathcal{K}\backslash\mathcal{C}} (p_{i}^{BS}-p_{i}^{NF})     =\sum_{i\in\mathcal{C}} (p_{i}^{NF}-p_{i}^{BS}) \nonumber\\& \Rightarrow
%
% \sum_{i\in\mathcal{K}\backslash\mathcal{C}} |p_{i}^{BS}-p_{i}^{NF}|     &\geq \sum_{i\in\mathcal{C}} (p_{i}^{NF}-p_{i}^{BS}) \Rightarrow\nonumber\\
\sum_{i\in\mathcal{K}\backslash\mathcal{C}} |p_{i}^{BS}-p_{i}^{NF}| \geq G \label{eq:derivation-bound-avg-step-2}
\end{align}
\end{footnotesize}%
where in the first equation both sums equal to 1 (probability distributions), the second equation follows by moving all terms for $i\in\mathcal{K}\backslash\mathcal{C}$ to the left hand side, and in the third equation the left hand side follows from the property $x\leq |x|$ and the right hand side directly from the definition of $G$ (\eq{eq:gain-definition}).

Substituting \eq{eq:derivation-bound-avg-step-2} in \eq{eq:derivation-bound-avg-step-1} gives
\begin{footnotesize}
\begin{equation*}
G \leq 2\cdot F_{tv} - G \hspace{0.5cm}\Rightarrow\hspace{0.5cm} G \leq F_{tv}
\end{equation*}
\end{footnotesize}%
\myitem{F-kl.} Due to the logarithm involved in the expression of $F_{kl}$, we cannot proceed similarly to the cases of $F_{max}$ or $F_{tv}$, and we calculate the bound by solving the optimization problem of \eq{eq:optimization-problem-bound} with the method of Lagrangian multipliers. We first formulate the Lagrangian function $\mathcal{L}$ as follows\footnote{The problem \eq{eq:optimization-problem-bound} involves also the constraints $0\leq p_{i}^{NF} \leq1$, $\forall i\in\mathcal{K}$, which need to be accounted in the Lagrangian. However, if $p_{i}^{NF}$ is 0 or 1, the $F_{kl}$ diverges and thus the constraint in \eq{eq:optimization-problem-bound} is not satisfied. Hence, for any feasible solution it will hold that $0< p_{i}^{NF} <1$ and the corresponding Lagrange multipliers will be equal to zero (Karush–Kuhn–Tucker conditions).}

\vspace{-\baselineskip}
\begin{footnotesize}
\begin{align*}
\mathcal{L}%(\mathbf{p^{NF}},\lambda, \mu) 
= \textstyle \sum_{i\in\mathcal{C}} (p_{i}^{NF}- p_{i}^{BS}) - \lambda \cdot (F_{kl}-c_{f}) - \mu \cdot \left(\sum_{i\in\mathcal{K}}p_{i}^{NF}-1\right)
\end{align*}
\end{footnotesize}
The derivative of $\mathcal{L}$ with respect to $p_{i}^{NF}$ is
\begin{footnotesize}
\begin{align}
\frac{\partial \mathcal{L}}{\partial p_{i}^{NF}} =\left\{
\begin{tabular}{ll}
$1+\lambda\cdot \frac{p_{i}^{BS}}{p_{i}^{NF}}-\mu$     & ,  $i\in\mathcal{C}$\\
$\lambda \cdot \frac{p_{i}^{BS}}{p_{i}^{NF}}-\mu$     & ,  $i\in\mathcal{K}\backslash\mathcal{C}$
\end{tabular}
\right.
\end{align}
\end{footnotesize}%
% where we calculated the derivative of $F_{kl}$ from the expression of Section~\ref{sec:fairness-definition} as $\frac{\partial F_{kl}}{\partial p_{i}^{NF}} = -\frac{p_{i}^{BS}}{p_{i}^{NF}}$.
where we calculate $\frac{\partial F_{kl}}{\partial p_{i}^{NF}} = -\frac{p_{i}^{BS}}{p_{i}^{NF}}$ (see $F_{kl}$ definition; Sec.~\ref{sec:fairness-definition}).
Setting $\frac{\partial \mathcal{L}}{\partial p_{i}^{NF}} =0$ for the optimal solution, gives:
\begin{footnotesize}
\begin{align}\label{eq:bound-kl-pi-parametric}
p_{i}^{NF} = \left\{
\begin{tabular}{ll}
$\frac{\lambda}{\mu-1} \cdot p_{i}^{BS}$ & ,  $i\in\mathcal{C}$\\
$\frac{\lambda}{\mu} \cdot p_{i}^{BS}$ & ,  $i\in\mathcal{K}\backslash\mathcal{C}$
\end{tabular}
\right.
\end{align}
\end{footnotesize}%
To calculate the Lagrange multipliers, we use the definition of $G$ (\eq{eq:gain-definition}), substitute from \eq{eq:bound-kl-pi-parametric}, and get

% \vspace{-\baselineskip}
% \begin{footnotesize}
% \begin{align}
% G = \sum_{i\in\mathcal{C}}\frac{\lambda}{\mu-1} p_{i}^{BS}-p_{i}^{BS}    \Rightarrow \nonumber\\
% \frac{\lambda}{\mu-1} = 1+\frac{G}{\sum_{i\in\mathcal{C}}p_{i}^{BS}} = 1+\frac{G}{H} \label{eq:bound-kl-parameter1}
% \end{align}
% \end{footnotesize}
%%%%%%%%%%%%%%%%%%%%%%%%%%%%%%
%%%%%%%%%%%%%%%%%%%%%%%%%%%%%%
%%%%%%%%%%%%%%%%%%%%%%%%%%%%%%
%%%%%%%%%%%%%%%%%%%%%%%%%%%%%%
%%%%%%%%%%%%%%%%%%%%%%%%%%%%%%
\vspace{-\baselineskip}
\begin{footnotesize}
\begin{align}
G = \sum_{i\in\mathcal{C}}\frac{\lambda}{\mu-1} p_{i}^{BS}-p_{i}^{BS} \Rightarrow
\frac{\lambda}{\mu-1} = 1+\frac{G}{\displaystyle\sum_{i\in\mathcal{C}}p_{i}^{BS}} = 1+\frac{G}{H} \label{eq:bound-kl-parameter1}
\end{align}
\end{footnotesize}%
where for brevity we denoted $H = CHR^{BS}= \sum_{i\in\mathcal{C}}p_{i}^{BS}$. Then we consider the constraint $\sum_{i\in\mathcal{K}}p_{i}^{NF}=1$ and substituting from the expressions in \eq{eq:bound-kl-pi-parametric} and \eq{eq:bound-kl-parameter1} we get

% \vspace{-\baselineskip}
% \begin{footnotesize}
% \begin{align}
%  \sum_{i\in\mathcal{C}} \left(1+\frac{G}{H}\right)\cdot p_{i}^{BS} + \sum_{i\in\mathcal{K}\backslash\mathcal{C}} \frac{\lambda}{\mu}\cdot p_{i}^{BS} &=1   \Rightarrow \nonumber \\
%  \left(1+\frac{G}{H}\right) \cdot \sum_{i\in\mathcal{C}} p_{i}^{BS} +  \frac{\lambda}{\mu}\cdot \sum_{i\in\mathcal{K}\backslash\mathcal{C}} p_{i}^{BS} &=1   \Rightarrow \nonumber \\
%  \left(1+\frac{G}{H}\right) \cdot H + \frac{\lambda}{\mu}\cdot (1-H) &=1   \Rightarrow \nonumber \\
% \frac{\lambda}{\mu} = 1-\frac{G}{1-H} \label{eq:bound-kl-parameter2}
% \end{align}
% \end{footnotesize}
%%%%%%%%%%%%%%%%%%%%
%%%%%%%%%%%%%%%%%%%%
%%%%%%%%%%%%%%%%%%%%
\vspace{-\baselineskip}
\begin{footnotesize}
\begin{align}
& \sum_{i\in\mathcal{C}} \left(1+\frac{G}{H}\right)\cdot p_{i}^{BS} + \sum_{i\in\mathcal{K}\backslash\mathcal{C}} \frac{\lambda}{\mu}\cdot p_{i}^{BS} =1   \Rightarrow \nonumber \\
& \left(1+\frac{G}{H}\right) \cdot \sum_{i\in\mathcal{C}} p_{i}^{BS} +  \frac{\lambda}{\mu}\cdot \sum_{i\in\mathcal{K}\backslash\mathcal{C}} p_{i}^{BS} =1   \Rightarrow \nonumber \\
& \left(1+\frac{G}{H}\right) \cdot H + \frac{\lambda}{\mu}\cdot (1-H) =1 ~~~~\Rightarrow~~~~
\frac{\lambda}{\mu} = 1-\frac{G}{1-H} \label{eq:bound-kl-parameter2}
\end{align}
\end{footnotesize}%
where we used $\sum_{i\in\mathcal{K}\backslash\mathcal{C}} p_{i}^{BS} = 1 - \sum_{i\in\mathcal{C}} p_{i}^{BS} = 1-H$.

Now, substituting from \eq{eq:bound-kl-pi-parametric}, \eq{eq:bound-kl-parameter1} and \eq{eq:bound-kl-parameter2}, in the expression for the $F_{kl}$, gives

\vspace{-\baselineskip}
\begin{footnotesize}
\begin{align}
&F_{kl}= \nonumber\\
&\sum_{i\in\mathcal{C}} p_{i}^{BS}\cdot\log\left(\frac{p_{i}^{BS}}{\left(1+\frac{G}{H}\right) p_{i}^{BS} }\right)
+ \sum_{i\in\mathcal{K}\backslash\mathcal{C}} p_{i}^{BS}\cdot\log\left(\frac{p_{i}^{BS}}{\left(1-\frac{G}{1-H}\right) p_{i}^{BS} }\right) \nonumber\\
&= -\sum_{i\in\mathcal{C}} p_{i}^{BS}\cdot\log\left(1+\frac{G}{H}\right)
- \sum_{i\in\mathcal{K}\backslash\mathcal{C}} p_{i}^{BS}\cdot\log\left(1-\frac{G}{1-H}\right) \nonumber\\
&= -H\cdot\log\left(1+\frac{G}{H}\right)
- (1-H)\cdot\log\left(1-\frac{G}{1-H}\right)
\end{align}
\end{footnotesize}
The above equality holds for the optimal $\mathbf{p^{NF}}$, i.e., the maximum network gain $G$; for any other $\mathbf{p^{NF}}$ the gains will be lower, which makes the above the inequality of Theorem~\ref{thm:bounds}.

% \clearpage
% \newpage
% \section{Temp Material}
% \input{temp}
% \section{Characterization Detailed BACKUP}
% % \input{characterization_detailed_BACKUP}


\begin{thebibliography}{10}
\providecommand{\url}[1]{#1}
\csname url@samestyle\endcsname
\providecommand{\newblock}{\relax}
\providecommand{\bibinfo}[2]{#2}
\providecommand{\BIBentrySTDinterwordspacing}{\spaceskip=0pt\relax}
\providecommand{\BIBentryALTinterwordstretchfactor}{4}
\providecommand{\BIBentryALTinterwordspacing}{\spaceskip=\fontdimen2\font plus
\BIBentryALTinterwordstretchfactor\fontdimen3\font minus
  \fontdimen4\font\relax}
\providecommand{\BIBforeignlanguage}[2]{{%
\expandafter\ifx\csname l@#1\endcsname\relax
\typeout{** WARNING: IEEEtran.bst: No hyphenation pattern has been}%
\typeout{** loaded for the language `#1'. Using the pattern for}%
\typeout{** the default language instead.}%
\else
\language=\csname l@#1\endcsname
\fi
#2}}
\providecommand{\BIBdecl}{\relax}
\BIBdecl

\bibitem{sch-chants-2016}
T.~Spyropoulos and P.~Sermpezis, ``Soft cache hits and the impact of
  alternative content recommendations on mobile edge caching,'' in \emph{Proc.
  ACM Workshop on Challenged Networks (CHANTS)}, 2016.

\bibitem{chatzieleftheriou2017caching}
L.-E. Chatzieleftheriou, M.~Karaliopoulos, and I.~Koutsopoulos, ``Caching-aware
  recommendations: Nudging user preferences towards better caching
  performance,'' in \emph{Proc. IEEE INFOCOM}, 2017.

\bibitem{sermpezis2018soft}
P.~Sermpezis, T.~Giannakas, T.~Spyropoulos, and L.~Vigneri, ``Soft cache hits:
  Improving performance through recommendation and delivery of related
  content,'' \emph{IEEE Journal Selected Areas in Communications}, 2018.

\bibitem{giannakas-wowmom-2018}
T.~Giannakas, P.~Sermpezis, and T.~Spyropoulos, ``Show me the cache: Optimizing
  cache-friendly recommendations for sequential content access,'' in
  \emph{Proc. IEEE WoWMoM}, 2018.

\bibitem{kastanakis-cabaret-mecomm-2018}
S.~Kastanakis, P.~Sermpezis, V.~Kotronis, and X.~Dimitropoulos, ``Cabaret:
  Leveraging recommendation systems for mobile edge caching,'' in \emph{Proc.
  ACM SIGCOMM workshops}, 2018.

\bibitem{kastanakis2020network}
S.~Kastanakis, P.~Sermpezis, V.~Kotronis, D.~S. Menasche, and T.~Spyropoulos,
  ``Network-aware recommendations in the wild: Methodology, realistic
  evaluations, experiments,'' \emph{IEEE Trans. on Mobile Comp.}, 2020.

\bibitem{zhu2018coded}
B.~Zhu and W.~Chen, ``Coded caching with joint content recommendation and user
  grouping,'' in \emph{Proc. IEEE GLOBECOM}, 2018.

\bibitem{chatzieleftheriou2019jointly}
L.~E. Chatzieleftheriou, M.~Karaliopoulos, and I.~Koutsopoulos, ``Jointly
  optimizing content caching and recommendations in small cell networks,''
  \emph{IEEE Trans. on Mobile Computing}, vol.~18, no.~1, 2019.

\bibitem{costantini2019approximation}
M.~Costantini, T.~Spyropoulos, T.~Giannakas, and P.~Sermpezis, ``Approximation
  guarantees for the joint optimization of caching and recommendation,'' in
  \emph{Proc. IEEE ICC}, 2020.

\bibitem{garetto2020similarity}
M.~Garetto, E.~Leonardi, and G.~Neglia, ``Similarity caching: Theory and
  algorithms,'' in \emph{Proc. IEEE INFOCOM}, 2020.

\bibitem{qi2018optimizing}
K.~Qi, B.~Chen, C.~Yang, and S.~Han, ``Optimizing caching and recommendation
  towards user satisfaction,'' in \emph{IEEE WCSP}, 2018.

\bibitem{giannakas2019order}
T.~Giannakas, T.~Spyropoulos, and P.~Sermpezis, ``The order of things:
  Position-aware network-friendly recommendations in long viewing sessions,''
  in \emph{Proc. WiOpt}, 2019.

\bibitem{chatzieleftheriou2019joint}
L.~Chatzieleftheriou, G.~Darzanos, M.~Karaliopoulos, and I.~Koutsopoulos,
  ``Joint user association, content caching and recommendations in wireless
  edge networks,'' \emph{PER}, vol.~46, no.~3, pp. 12--17, 2019.

\bibitem{gupta2019effect}
S.~Gupta and S.~Moharir, ``Effect of recommendations on serving content with
  unknown demand,'' \emph{ACM TOMPECS}, vol.~4, no.~1, p.~4, 2019.

\bibitem{lin2018joint}
Z.~Lin and W.~Chen, ``Joint pushing and recommendation for susceptible users
  with time-varying connectivity,'' in \emph{IEEE GLOBECOM}, 2018.

\bibitem{song2018making}
L.~Song and C.~Fragouli, ``Making recommendations bandwidth aware,'' \emph{IEEE
  Trans. Information Theory}, vol.~64, no.~11, 2018.

\bibitem{lin2019content}
Z.~Lin and W.~Chen, ``Content pushing over multiuser miso downlinks with
  multicast beamforming and recommendation: A cross-layer approach,''
  \emph{IEEE Trans. on Communications}, vol.~67, no.~10, 2019.

\bibitem{giannakas2020soba}
T.~Giannakas, A.~Giovanidis, and T.~Spyropoulos, ``Soba: Session optimal
  mdp-based network friendly recommendations,'' in \emph{Proc. IEEE INFOCOM},
  2021.

\bibitem{sermpezis2019towards}
P.~Sermpezis, S.~Kastanakis, J.~I. Pinheiro, F.~Assis, D.~Menasch{\'e}, and
  T.~Spyropoulos, ``Towards qos-aware recommendations,'' in \emph{ACM RecSys
  workshops (CARS workshop)}, 2020.

\bibitem{cache-centric-video-recommendation}
D.~Krishnappa, M.~Zink, C.~Griwodz, and P.~Halvorsen, ``Cache-centric video
  recommendation: an approach to improve the efficiency of youtube caches,''
  \emph{ACM Transactions on Multimedia Computing, Communications, and
  Applications (TOMM)}, vol.~11, no.~4, p.~48, 2015.

\bibitem{content-recommendation-swarming}
D.~Munaro, C.~Delgado, and D.~S. Menasch{\'e}, ``Content recommendation and
  service costs in swarming systems,'' in \emph{Proc. IEEE ICC}, 2015.

\bibitem{cisco2018}
{Cisco}, ``Visual networking index: Forecast and trends, 2017-2022,'' 2018.

\bibitem{ericsson2018}
{Ericsson}, ``Ericsson mobility report,'' 2018, white paper.

\bibitem{RecImpact-IMC10}
R.~Zhou, S.~Khemmarat, and L.~Gao, ``The impact of youtube recommendation
  system on video views,'' in \emph{Proc. ACM IMC}, 2010.

\bibitem{gomez2016netflix}
C.~Gomez-Uribe and N.~Hunt, ``The netflix recommender system: Algorithms,
  business value, and innovation,'' \emph{ACM Transactions on Management
  Information Systems (TMIS)}, vol.~6, no.~4, p.~13, 2016.

\bibitem{abdollahpouri2019multi}
H.~Abdollahpouri and R.~Burke, ``Multi-stakeholder recommendation and its
  connection to multi-sided fairness,'' in \emph{Proc. RMSE workshop at ACM
  RecSys}, 2019.

\bibitem{burke2017multisided}
R.~Burke, ``Multisided fairness for recommendation,'' in \emph{Workshop on
  Fairness, Accountability, Transparency in Machine Learning}, 2017.

\bibitem{burke2018balanced}
R.~Burke, N.~Sonboli, and A.~Ordonez-Gauger, ``Balanced neighborhoods for
  multi-sided fairness in recommendation,'' in \emph{Conference on Fairness,
  Accountability and Transparency}, 2018, pp. 202--214.

\bibitem{edizel2020fairecsys}
B.~Edizel, F.~Bonchi, S.~Hajian, A.~Panisson, and T.~Tassa, ``Fairecsys:
  Mitigating algorithmic bias in recommender systems,'' \emph{International
  Journal of Data Science and Analytics}, vol.~9, no.~2, pp. 197--213, 2020.

\bibitem{patro2020incremental}
G.~K. Patro, A.~Chakraborty, N.~Ganguly, and K.~Gummadi, ``Incremental fairness
  in two-sided market platforms: On smoothly updating recommendations,'' in
  \emph{Proc. AAAI conf. on Artificial Intelligence}, 2020.

\bibitem{pessach2020algorithmic}
D.~Pessach and E.~Shmueli, ``Algorithmic fairness,'' \emph{arXiv preprint
  arXiv:2001.09784}, 2020.

\bibitem{Sacharidis2019ACA}
D.~Sacharidis, K.~Mouratidis, and D.~Kleftogiannis, ``A common approach for
  consumer and provider fairness in recommendations,'' in \emph{Proc. ACM
  RecSys (Late-breaking Results,)}, 2019.

\bibitem{steck2018calibrated}
H.~Steck, ``Calibrated recommendations,'' in \emph{Proc. ACM RecSys}, 2018.

\bibitem{yang2017measuring}
K.~Yang and J.~Stoyanovich, ``Measuring fairness in ranked outputs,'' in
  \emph{Proc. SSDBM}, 2017.

\bibitem{liu2018personalizing}
W.~Liu and R.~Burke, ``Personalizing fairness-aware re-ranking,'' in
  \emph{Proc. FATREC workshop at ACM RecSys}, 2018.

\bibitem{sarwar2001item}
B.~Sarwar, G.~Karypis, J.~Konstan, and J.~Riedl, ``Item-based collaborative
  filtering recommendation algorithms,'' in \emph{Proc. ACM WWW}, 2001.

\bibitem{covington2016deep}
P.~Covington, J.~Adams, and E.~Sargin, ``Deep neural networks for {Y}ou{T}ube
  recommendations,'' in \emph{Proc. ACM RecSys}, 2016.

\bibitem{mehrotra2018towards}
R.~Mehrotra, J.~McInerney, H.~Bouchard, M.~Lalmas, and F.~Diaz, ``Towards a
  fair marketplace: Counterfactual evaluation of the trade-off between
  relevance, fairness \& satisfaction in recommendation systems,'' in
  \emph{Proc. ACM CIKM}, 2018.

\bibitem{lan2010axiomatic}
T.~Lan, D.~Kao, M.~Chiang, and A.~Sabharwal, ``An axiomatic theory of fairness
  in network resource allocation,'' in \emph{IEEE INFOCOM}, 2010.

\bibitem{lastfm-related-content-dataset}
``Last.fm dataset,'' \url{https://labrosa.ee.columbia.edu/millionsong/lastfm}.

\bibitem{movielens-related-dataset}
F.~M. Harper and J.~A. Konstan, ``The movielens datasets: History and
  context,'' \emph{ACM Trans. on Interactive Intelligent Systems (TiiS)}, 2016.

\bibitem{langville2004deeper}
A.~N. Langville and C.~D. Meyer, ``Deeper inside pagerank,'' \emph{Internet
  Mathematics}, vol.~1, no.~3, pp. 335--380, 2004.

\bibitem{xiao2017fairness}
L.~Xiao, Z.~Min, Z.~Yongfeng, G.~Zhaoquan, L.~Yiqun, and M.~Shaoping,
  ``Fairness-aware group recommendation with pareto-efficiency,'' in
  \emph{Proc. ACM RecSys}, 2017.

\end{thebibliography}
\end{document}